\documentclass[final, onefignum,onetabnum]{siamonline250211}

\usepackage[utf8]{inputenc}    
\usepackage[T1]{fontenc}

\def\bN{{\mathbb N}}

\def\bR{{\mathbb R}}

\def\X{{\mathcal{X}}}
\def\Y{{\mathcal{Y}}}

\usepackage{lipsum}
\usepackage{amsfonts}
\usepackage{graphicx}
\usepackage{epstopdf}
\usepackage{algorithmic}
\ifpdf
  \DeclareGraphicsExtensions{.eps,.pdf,.png,.jpg}
\else
  \DeclareGraphicsExtensions{.eps}
\fi

\usepackage{enumitem}
\setlist[enumerate]{leftmargin=.5in}
\setlist[itemize]{leftmargin=.5in}

\newsiamremark{remark}{Remark}
\newsiamremark{hypothesis}{Hypothesis}
\crefname{hypothesis}{Hypothesis}{Hypotheses}
\newsiamthm{claim}{Claim}
\newsiamremark{fact}{Fact}
\crefname{fact}{Fact}{Facts}

\newtheorem{theo}{Theorem}

{\tiny }%

\headers{Joint Uncertainty Quantification of ODE Parameter and Discretization Error}{S.~Toyota and Y.~Miyatake}

\title{Joint Bayesian Inference of Parameter and Discretization Error Uncertainties in ODE Models %\thanks{
}

\author{Shoji Toyota\thanks{Department of Advanced Information Technology, Kyushu University,  Japan 
  (\email{toyota@ait.kyushu-u.ac.jp}).}
\and Yuto Miyatake\thanks{D3 Center, The University of Osaka, Japan
  (\email{yuto.miyatake.cmc@osaka-u.ac.jp}).}}

\usepackage{amsopn}

\ifpdf
\hypersetup{
  pdftitle={Joint Bayesian Inference of Parameter and Discretization Error Uncertainties in
ODE Models},
  pdfauthor={S.~Toyota and Y.~Miyatake}
}
\fi

\usepackage{graphicx}
\usepackage{subcaption} 
\usepackage{wrapfig}
\usepackage{amsmath}
\usepackage{booktabs}
\usepackage{lipsum}
\usepackage{diagbox}
\usepackage{placeins}
\usepackage{afterpage}
\usepackage{comment}
\usepackage{longtable}

\usepackage{tikz}
\usetikzlibrary{positioning,arrows.meta}

\DeclareMathOperator{\categorical}{Categorical}

\begin{document}

\maketitle

\begin{abstract}
We address the problem of Bayesian inference for parameters in ordinary differential equation (ODE) models based on observational data. 
Conventional approaches in this setting typically rely on numerical solvers such as the Euler or Runge–Kutta methods. However, these methods generally do not account for the discretization error induced by discretizing the ODE model. 
We propose a Bayesian inference framework for ODE models that explicitly quantifies discretization errors. Our method models discretization error as a random variable and performs Bayesian inference on both ODE parameters and variances of the randomized discretization errors, referred to as the {\em discretization error variance.} 
A key idea of our approach is the introduction of a Markov prior on the temporal evolution of the discretization error variances, enabling the inference problem to be formulated as a state-space model.
Furthermore, we propose a specific form of the Markov prior that arises naturally from standard discretization error analysis. 
This prior depends on the step size $h$ in the numerical solver, and we discuss its asymptotic property as $h \rightarrow +0$. 
Numerical experiments illustrate that the proposed method can simultaneously quantify uncertainties in both the ODE parameters and the discretization errors, and can produce posterior distributions over the parameters with broader support by accounting for discretization error.

\end{abstract}

% REQUIRED

% REQUIRED

\section{Introduction}
Ordinary differential equation (ODE) models are widely used in various scientific and engineering disciplines, including neuroscience \cite{fitzhugh1961impulses}, epidemiology \cite{kermack1927contribution}, population biology \cite{Volterra1926FluctuationsIT} and control engineering \cite{ogata2010modern}. These models often involve parameters that cannot be directly observed or determined; inferring such parameters from observational data is a fundamental problem. 
Bayesian inference is a standard approach for parameter estimation in ODE models; by placing a distribution---referred to as a \textit{prior}---over the ODE parameter space, Bayesian inference allows us to estimate parameters in the form of a distribution known as a \textit{posterior}. Since this posterior distribution is a probability distribution, Bayesian inference enables parameter inference along with uncertainty quantification.

One of the primary challenges in performing Bayesian inference for ODE models is the intractability of their likelihoods
as they usually involve the exact solution to the ODE models. 
Note that standard Bayesian inference methods, such as Markov Chain Monte Carlo \cite{gilks1995markov} or variational methods \cite{beal2003}, typically assume that model likelihoods are tractable; thus, this assumption is often violated in the context of ODE models. 
A practical strategy to address this issue involves using numerical solvers, such as Euler or Runge–Kutta methods, to approximate solutions to the ODEs \cite{butcher2016numerical, hairer2020solving}. 
For instance, as often employed in simulation-based inference schemes \cite{cranmer2020frontier} and 4D variational (4D-Var) data assimilation \cite{asch2016data}, the likelihood is evaluated approximately by replacing the exact solution in the likelihood with a numerical solution. Another widely used method from statistics is Approximate Bayesian Computation (ABC) \cite{beaumont2002approximate}, which estimates the posterior distribution by comparing observations to numerical solutions obtained from numerical solvers.

These Bayesian inference methods for ODEs typically assume that numerical solvers yield sufficiently accurate solutions.
However, this assumption does not always hold.
For example, for chaotic systems, large-scale problems, and highly oscillatory dynamics, obtaining accurate numerical solutions remains challenging.
In such settings, discretization errors may lead to biased or misleading posterior distributions (see, e.g., \cite{conrad2017statistical}).
This highlights the need for frameworks that can quantify the uncertainty arising from discretization errors in ODEs and to infer ODE parameters that account for this uncertainty.

Motivated by this objective, various methods for quantifying discretization errors have been developed, particularly within the framework of \emph{Probabilistic Numerics}~\cite{hennig2022probablistic}.  
The central idea of the framework is to interpret numerical computation as a statistical inference, enabling a principled quantification of discretization uncertainty.  
A well-established example
is Bayesian ODE solvers~\cite{beck2024diffusiontemper,  kersting2020,le2025modelling, schmidt2021sir, SchoberSarkkaHennig2018, pmlr-v162-tronarp22a, Tronarp2018, TroSarHen2021},  
which reformulate the numerical solution of ODEs as a Gaussian process inference problem, in contrast to traditional deterministic solvers.
Another prominent line is perturbative approaches~\cite{abdulle2020random, conrad2017statistical, lie2022randomised, lie2019strong},  
which compute numerical solutions together with probabilistic perturbations.

A distinct class of discretization error quantification methods for ODEs is the \emph{discretization error variance} approach~\cite{MARUMO2024modelling, Matsuda2019estimation, miyatake2025quantifying}. In this framework, standard numerical solvers such as the Runge--Kutta method are employed to solve ODEs, and the discretization errors introduced by these solvers are modeled as random variables. Their variances---referred to as \emph{discretization error variances}---are treated as statistical quantities to be inferred from observations.  
A key distinction between this approach and those described above is that the discretization error itself is inferred directly from observations. Another important feature is that insights from classical error analysis can be incorporated into the modeling of the discretization error variance.  
For example, in~\cite{MARUMO2024modelling, Matsuda2019estimation, miyatake2025quantifying}, a monotonicity constraint is imposed on the discretization error variance, reflecting the well-known principle from numerical analysis that numerical errors tend to accumulate over time.  
However, within this context, a Bayesian framework that simultaneously infers both the model parameters and the discretization error variance remains undeveloped.

In this study, we propose a Bayesian inference framework for jointly estimating ODE parameters and discretization errors, building upon the discretization error variance approach. 
By adopting a fully Bayesian perspective, we place priors not only on the ODE parameters but also on the discretization error variances, thereby enabling joint uncertainty quantification of both ODE parameters and discretization errors.

A key feature of our framework is to set a prior over the time course of the discretization error variances that satisfies the Markov property. By imposing a  Markov property on a prior, our objectives can be formulated as an inference on state-space models, in which the discretization error variances, observations, and ODE parameters correspond to the latent variables, observation variables, and parameters of the state-space model, respectively (Figure~\ref{fig:state_space}). In state-space modeling, various methodologies have been developed for estimating the posterior distributions of latent states together with parameters in a state-space model; by leveraging them, we are able to perform joint Bayesian inference on both the discretization error variances and the ODE parameters. In particular, we employ particle filtering \cite{doucet2001sequential,gordon1993novel,kitagawa1996monte}, in combination with a self-organizing technique \cite{kitagawa1998self-organizing}, to obtain the posterior over discretization error variances and ODE parameters.

\begin{figure}[t]
    \centering
    \begin{tikzpicture}
        \node[circle, draw, text width=0.9cm,align=center] (c1) {$\Sigma_{t_{i-1}}$};
         \node[circle, draw, text width=0.9cm,align=center, right = 1.2cm of c1] (c2) {$\Sigma_{t_i}$};
         \node[circle, draw, text width=0.9cm,align=center, right = 1.2cm of c2] (c3) {$\Sigma_{t_{i+1}}$};
         \node[circle, draw, text width=0.9cm,align=center, above = 1cm of c1] (c4) {$y_{t_{i-1}}$};
         \node[circle, draw, text width=0.9cm,align=center, above = 1cm of c2] (c5) {$y_{t_i}$};
         \node[circle, draw, text width=0.9cm,align=center, above = 1cm of c3] (c6) {$y_{t_{i+1}}$};
         \node[left = 1.2cm of c1] (c0) {};
         \node[right = 1.2cm of c3] (c7) {};

         \draw [-{Stealth[length=2mm]}, thick] (c1)--node[above] {$\theta$} (c2);
         \draw [-{Stealth[length=2mm]}, thick] (c2)--node[above] {$\theta$}(c3);
         \draw [-{Stealth[length=2mm]}, thick] (c1)--node[right]{$\theta$}(c4);
         \draw [-{Stealth[length=2mm]}, thick] (c2)--node[right]{$\theta$}(c5);
         \draw [-{Stealth[length=2mm]}, thick] (c3)--node[right]{$\theta$}(c6);

         \draw [-{Stealth[length=2mm]}, thick] (c0)--node[above] {$\theta$}(c1);
         \draw [-{Stealth[length=2mm]}, thick] (c3)--node[above] {$\theta$}(c7);
    \end{tikzpicture}
    \caption{Our state-space modeling formulation. Let $y_t$, $\theta$, and $\Sigma_t$ denote the observation, the ODE parameter, and the discretization error variance, respectively  (their formal definitions are provided in the next section).  By placing a Markov prior on the discretization error variances $\Sigma_{t_i}$,  we obtain a state-space model in which the latent variables correspond to the discretization error variances. The ODE parameters $\theta$ govern both the latent transition  $\Sigma_{t_i} \rightarrow \Sigma_{t_{i+1}}$  and the observation process $\Sigma_{t_i} \rightarrow y_{t_i}$. } 
    \label{fig:state_space}
\end{figure}
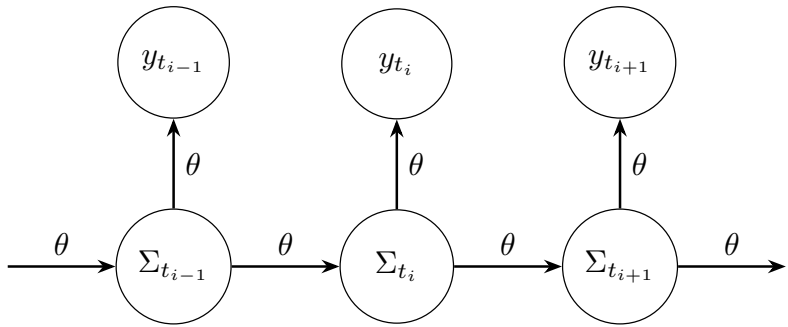

The specification of a Markov prior over the discretization error variances is crucial for ensuring the effectiveness of the proposed method. In this paper, we propose a Markov prior inspired by the standard discretization error analysis of ordinary differential equations (ODEs). In numerical analysis, it is well known that the discretization error can be represented by the accumulation of local  errors, the errors incurred at each individual step. From this viewpoint, we introduce a Markov prior in which the local truncation errors accumulate additively over time. We further investigate theoretical properties of the proposed prior; in particular, since the prior depends on the step size of the solver used to obtain the numerical solution, we analyze the asymptotic behavior of the proposed prior as the step size approaches zero.

The paper is organized as follows. Section~\ref{sec:preliminaries} provides background for this research. In Section~\ref{sec:method}, we present the proposed method. Specifically, we first reformulate the inference of discretization error variances as a state-space modeling problem (Subsection~\ref{subsec:state-space}). We then extend this framework to the joint Bayesian inference of both discretization error variances and ODE parameters using a self-organizing state-space model (Subsection~\ref{subsec:JointBayes}). In Section~\ref{sec:prior}, we construct a Markov prior for discretization error variances motivated by principles from numerical error analysis and establish an asymptotic property of the proposed prior as the step size goes to zero. Section~\ref{sec:experiments} demonstrates the effectiveness of our method through numerical experiments on {the chemical reaction network,} the pendulum system and the FitzHugh–Nagumo model. Finally, Section~\ref{sec:discussions} summarizes our contributions and discusses directions for future work.

\section{Preliminaries}
\label{sec:preliminaries}

\subsection{Bayesian Inference for ODE Models}

Consider a $d_\mathcal{X}$-dimensional ordinary differential equation (ODE):
\begin{equation}\label{eq:ODEmodel}
    \frac{dx(t)}{dt} = f(x(t; \theta), \theta),
\end{equation}
where the initial condition $x(0; \theta)$ and the solution $x(t; \theta)$ lie in $\mathbb{R}^{d_\X}$ with ${d_\X} \in \bN_{>0}$. 
Here $\theta \in \mathbb{R}^{d}$ is an unknown parameter that should be estimated. The function $f$ is defined as $f: \mathbb{R}^{d_\mathcal{X}} \times \mathbb{R}^d \rightarrow \mathbb{R}^{d_\mathcal{X}}$. We assume that noisy observations $y^*_{t_0}, \dots, y^*_{t_N} \in \mathbb{R}^{d_\Y}$ of the true trajectory $x(t; \theta^*)$ are available at discrete time points $ 0 \leq  t_0 < \cdots < t_N$, where
\begin{equation}\label{eq:observation_process}
    y^*_{t_i} = H  x\left( t_i; \theta^* \right) + \varepsilon_i. 
\end{equation}
Here, the linear observation operator $H \in  \mathbb{R}^{d_\mathcal{Y} \times d_\mathcal{X}}$ is a full rank matrix, and observation noises $\varepsilon_i \sim \mathcal{N}(0, \Gamma)$ are drawn i.i.d.~from a multivariate normal distribution with an invertible covariance matrix $\Gamma$. 

A principled approach to inferring the ODE parameter from observations $y^*_{t_0}, \dots, y^*_{t_N} $ is Bayesian inference. Given a likelihood $p\left(y_{t_0}, \dots, y_{t_N} \mid \theta \right)$ and a prior distribution $p(\theta)$, Bayesian inference aims to compute the posterior distribution:
$$
p(\theta \mid y^*_{t_0}, \dots, y^*_{t_N} ) \propto p( y^*_{t_0}, \dots, y^*_{t_N} \mid \theta) \cdot p(\theta).
$$
Traditional Bayesian inference methods, such as Markov Chain Monte Carlo (MCMC) \cite{gilks1995markov} and variational inference \cite{beal2003}, require access to an analytical form of the likelihood:
\begin{align*}
    &p\left(y_{t_0}, \dots, y_{t_N} \mid \theta \right) 
    = \prod_{i=0}^N p \left( y_{t_{i}} \mid \theta \right) \\
    &\quad = \prod_{i=0}^N 
    \frac{1}{\sqrt{(2\pi)^{d_{\mathcal{Y}}} |\Gamma|}}
    \exp\left\{
        -\frac{1}{2} \left( y_{t_{i}} - H x\left( t_i; \theta \right) \right)^{\mathsf{T}}
        \Gamma^{-1}
        \left( y_{t_{i}} - H x\left( t_i; \theta \right) \right)
    \right\},
\end{align*}
where $x(t_i; \theta)$ denotes the solution of the ODE at time $t_i$ given ODE parameter $\theta$.

However, evaluating this likelihood is often intractable as the exact solution $x(t_i; \theta)$ is inaccessible.
Although numerical solvers such as Runge–Kutta or Euler methods are commonly employed,
numerical solutions $x_{t_i}(\theta) \approx x(t_i; \theta)$ inevitably deviate from the true trajectory $x(t_i; \theta)$ due to discretization errors.
These errors introduce bias into the likelihood, which can significantly affect inference accuracy.
To address this issue, several methods have been proposed to statistically quantify discretization errors. Well-known examples include Bayesian ODE solvers~\cite{beck2024diffusiontemper, bosch2024parallel, kersting2020,  le2025modelling,schmidt2021sir,  SchoberSarkkaHennig2018, pmlr-v162-tronarp22a, Tronarp2018, TroSarHen2021} and perturbative methods~\cite{abdulle2020random, conrad2017statistical, lie2022randomised, lie2019strong}, both of which fall within the category of \emph{Probabilistic Numerics}~\cite{hennig2022probablistic}.

\subsection{Discretization Error Variance}

In the next section, we present a Bayesian inference method for estimating the ODE parameter $\theta$ that explicitly accounts for discretization errors.
Our discussion builds on the recently proposed concept of \emph{discretization error variance}~\cite{MARUMO2024modelling, Matsuda2019estimation, miyatake2025quantifying}.
The central idea of this concept is summarized as follows.

In contrast to Bayesian ODE solvers and perturbative methods, we explicitly model the discretization errors induced by a numerical solver as random variables.
Specifically, the discretization error at $t=t_i$:
\[
  r_{t_i} = x(t_i; \theta) - x_{t_i}(\theta)
\]
is assumed to follow a Gaussian distribution
\begin{equation}\label{eq:discretization_error_variances}
    r_{t_i} \sim \mathcal{N}(0, \Sigma_{t_i}).
\end{equation}
Here, the covariance matrix $\Sigma_{t_i}$, referred to as the \emph{discretization error variance}, is treated as a statistical quantity to be inferred from observations and numerical solutions, thereby providing a way to quantify the discretization error. 
Equivalently, the exact solution $x(t_i; \theta)$ can also be interpreted as a random variable:
\begin{equation}
    x(t_i; \theta) \sim \mathcal{N}(x_{t_i} (\theta), \Sigma_{t_i}),
\end{equation}
where $x_{t_i} (\theta)$ denotes a numerical approximation.

Under this framework, we obtain a modified likelihood
\begin{align}
    & p\left( y_{t_0}, \dots, y_{t_N} \mid \theta, \Sigma_{t_0}, \dots, \Sigma_{t_N} \right) = \prod_{i=0}^N p(y_{t_i} \mid \Sigma_{t_i},\theta ) \notag \\
    &\quad = \prod_{i=0}^N \frac{1}{\sqrt{(2\pi)^{d_{\mathcal{Y}}} |\Gamma + V_{t_i}|}}
\exp\left\{
 -\frac{1}{2} \left( y_{t_i} - Hx_{t_i} \left(\theta \right)  \right) ^{\mathsf{T}}
 (\Gamma + V_{t_i})^{-1}
 \left( y_{t_i} - H x_{t_i}\left( \theta \right) \right)
\right\} \label{eq:likelihood_with_discretization}
\end{align}
with $V_{t_i} = H \cdot \Sigma_{t_i}\cdot  H^\mathsf{T}$, which will also be referred to as the discretization error variance.
A central feature of this framework is to infer the ODE parameters and the discretization error variances simultaneously.

\section{
A New Approach for Joint Bayesian Inference of ODE Parameters and Discretization Error Variances} \label{sec:method}

In this study, we aim to develop a method for performing Bayesian inference for ODE parameters while accounting for discretization errors introduced by numerical solvers.
Specifically,
we focus on the ODE parameters \( \theta \) and the discretization error variances \( \Sigma_{t_i} \).

Given a prior distribution \( p(\theta, \Sigma_{t_0}, \dots, \Sigma_{t_N}) \) on the ODE parameter $\theta$ and discretization error variances $\Sigma_{t_i}$, our target is to obtain the posterior distribution
\begin{align}
    & p\left( \theta, \Sigma_{t_0}, \dots, \Sigma_{t_N} \mid y_{t_0}, \dots, y_{t_N} \right)  \propto p\left( y_{t_0}, \dots, y_{t_N} \mid \theta, \Sigma_{t_0}, \dots, \Sigma_{t_N} \right) \cdot p(\theta, \Sigma_{t_0}, \dots, \Sigma_{t_N}). \label{eq:our_goal}
\end{align}
The key components of this framework are the specification of the prior  \( p(\theta, \Sigma_{t_0}, \dots, \Sigma_{t_N}) \) and the construction of an algorithm for sampling from the posterior distribution.
In this section, we outline the general ideas underlying the prior and the sampling scheme.
A specific form of the prior will be discussed in the next section, noting that some flexibility in its choice remains.

We note that this form of joint Bayesian inference has not been explored in the literature on discretization error variance.
Previous studies either perform point estimation of the ODE parameters and discretization error variances~\cite{MARUMO2024modelling, Matsuda2019estimation}, or conduct Bayesian inference only for discretization error variances, without jointly inferring the ODE parameters~\cite{miyatake2025quantifying}.

\subsection{State-Space Model Interpretation of Discretization Error Variance Inference}\label{subsec:state-space}

To illustrate our core idea clearly, we first focus on the Bayesian inference of the discretization error variances $\Sigma_{t_i}$ alone.
In this setting, {instead} of the original objective~\eqref{eq:our_goal}, our target is the posterior distribution
\begin{align}
    & p\left(\Sigma_{t_0}, \dots, \Sigma_{t_N} \mid y_{t_0}, \dots, y_{t_N} \right)  \propto p\left( y_{t_0}, \dots, y_{t_N} \mid \Sigma_{t_0}, \dots, \Sigma_{t_N} \right) \cdot p(\Sigma_{t_0}, \dots, \Sigma_{t_N}), \label{eq:our_semigoal}
\end{align}
which concerns only 
the discretization error variances $\left( \Sigma_{t_0}, \dots, \Sigma_{t_N} \right)$. 
The joint inference (\ref{eq:our_goal}) of ODE parameters and discretization error variances will be discussed in Section~\ref{subsec:JointBayes}.

In our approach, we interpret the discretization error variances as latent variables and the observations $y_i$ as observation variables in a corresponding state-space model. Under this formulation, posterior sampling can be performed using a filtering-based approach.\footnote{The introduction of a Markov prior is also a key component in Bayesian ODE solvers for reducing computational costs \cite{beck2024diffusiontemper, bosch2024parallel, cockayne2019bayesian, schmidt2021sir, pmlr-v162-tronarp22a}.  
The essential difference, however, is that Bayesian ODE solvers place a Markov prior on solutions of the differential equation itself, whereas in this study the Markov prior is placed on trajectories of the discretization errors.  
As will be discussed in the next section, imposing a prior distribution directly on the discretization error allows insights from classical error analysis to be incorporated into the prior---something that is difficult to achieve when the prior is placed on the solution itself.

}

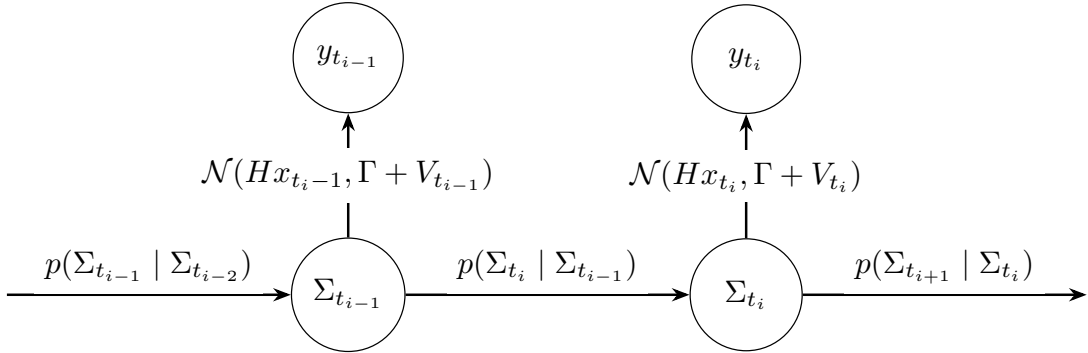
\begin{figure}[t]
    \centering

    \begin{tikzpicture}
    \node[circle, draw, text width=0.9cm,align=center] (c1) {$\Sigma_{t_{i-1}}$};
    \node[circle, draw, text width=0.9cm,align=center, right = 3.2cm of c1] (c2) {$\Sigma_{t_i}$};
    \node[circle, draw, text width=0.9cm,align=center, above = 1.4cm of c1] (c3) {$y_{t_{i-1}}$};
    \node[circle, draw, text width=0.9cm,align=center, above = 1.4cm of c2] (c4) {$y_{t_i}$};
    \node[left = 3.2cm of c1] (c5) {};
    \node[right = 3.2cm of c2] (c6) {};

    \draw [-{Stealth[length=2mm]}, thick] (c1)--node[above, fill=white,sloped]{$p(\Sigma_{t_i}\mid \Sigma_{t_{i-1}})$} (c2);
    \draw [-{Stealth[length=2mm]}, thick] (c1)--node[fill=white,sloped,rotate=-90]{$\mathcal{N}(H x_{t_i-1},\Gamma + V_{t_{i-1}})$} (c3);
    \draw [-{Stealth[length=2mm]}, thick] (c2)--node[fill=white,sloped,rotate=-90]{$\mathcal{N}(H x_{t_i},\Gamma + V_{t_i})$} (c4);
    \draw [-{Stealth[length=2mm]}, thick] (c5)--node[above, fill=white,sloped]{$p(\Sigma_{t_{i-1}}\mid \Sigma_{t_{i-2}})$}(c1);
    \draw [-{Stealth[length=2mm]}, thick] (c2)--node[above, fill=white,sloped]{$p(\Sigma_{t_{i+1}}\mid \Sigma_{t_i})$}(c6);
    \end{tikzpicture}
    
    \caption{Reduction of our objective (\ref{eq:our_semigoal}) to a state-space model. Its latent transition and observation process are given by a Markov prior $p(\Sigma_{t_{i+1}}|\Sigma_{t_{i}})$ and $\mathcal{N}(Hx_{t_i}, V_{t_i} + \Gamma)$ respectively.}
    \label{fig:state_space2}
\end{figure}

Specifically, we assume a Markovian structure for the prior distribution over the discretization error variances:
\[
\Sigma_{t_0} \to \Sigma_{t_1} \to \cdots \to \Sigma_{t_N}.
\]
Under this Markov property, the discretization error variances and the observations form a state-space model, as illustrated in Figure~\ref{fig:state_space2}.
Here, the latent transition $\Sigma_{t_i} \rightarrow \Sigma_{t_{i+1}}$ corresponds to the Markov prior $p\bigl( \Sigma_{t_{i+1}} \mid \Sigma_{t_i} \bigr)$ over the discretization error variances, while the observation process $\Sigma_{t_i} \rightarrow y_{t_i}$ is given by
\begin{equation}
    p\bigl( y_{t_i} \mid \Sigma_{t_i} \bigr) 
     = \frac{1}{\sqrt{(2\pi)^{d_{\mathcal{Y}}} \lvert \Gamma + V_{t_i} \rvert}}
    \exp\left\{
        -\frac{1}{2} \left( y_{t_i} - H x_{t_i}  \right)^{\mathsf{T}}
        (\Gamma + V_{t_i})^{-1}
        \left( y_{t_i} -  H x_{t_i} \right)
    \right\},
    \label{eq:ml_ySigma}
\end{equation}
which is a multivariate normal distribution with mean given by the numerical solution $ Hx_{t_i}$ at time $t_i$ and covariance $\Gamma + V_{t_i}$. Here, $V_{t_i}$ is defined by $ V_{t_i}= H \cdot \Sigma_{t_i}\cdot  H^\mathsf{T}$.

For the posterior sampling, 
we primarily employ the particle filter \cite{ doucet2001sequential,gordon1993novel, kitagawa1996monte}, as detailed in Algorithm~\ref{alg:particle}. 
Hereafter, the tuples $\left\{ y_{t_0}, ..., y_{t_i}  \right\}$ and $\left\{ \Sigma_{t_0}, ..., \Sigma_{t_i} \right\}$ are often abbreviated by $y_{0:i}$ and $\Sigma_{{0:i}}$ respectively.  
The particle filter approximates the posterior distribution $p\left( \Sigma_{t_{i}} \mid y_{0:i}^* \right)$ by a finite sum $\frac{1}{K} \sum_{k=1}^K \delta_{ \Sigma_{i|0:i}^k}
$,
where $\delta_{ \Sigma_{i|0:i}^k}$ denotes the Dirac measure at the discrete point $\Sigma_{i|0:i}^k $, 
and  each $\Sigma_{i|0:i}^k $ is called a \emph{particle}.
Given $K$ particles $\{\Sigma_{i|0:i}^k \}_{k=1}^K$ representing $p\left( \Sigma_{t_{i}} \mid y_{0:i}^* \right)$, the particles $\{\Sigma_{i+1|0:i+1}^k \}_{k=1}^K$ of $p\left( \Sigma_{t_{i+1}} \mid y_{0:i+1}^* \right)$ at the next time step $t_{i+1}$ are recursively obtained through the \emph{prediction} and \emph{correction} steps (see Sections~\ref{subsubsec:prediction} and~\ref{subsubsec:correction}). 
Finally, particles representing the target distribution $p\left( \Sigma_{{0:N}} \mid y_{0:N}^* \right)$ are obtained via the \emph{smoothing} step (see Section~\ref{subsubsec:smoothing}).

Before detailing each step, we emphasize that while evaluating the modified likelihood \eqref{eq:ml_ySigma} is straightforward, it is crucial to choose the Markov prior $p\left( \Sigma_{t_{i+1}} \mid \Sigma_{t_i} \right)$ such that sampling from it can be performed efficiently (the selection will be discussed in Section~\ref{sec:prior}).

\begin{figure}[tp]
\centering
\begin{algorithm}[H]
\caption{Particle Filter for Bayesian Inference on $\Sigma_{t_i}$}\label{alg:particle}
\begin{algorithmic}[1]
\STATE \textbf{Input:} A prior $ p(\Sigma_{t_{0}})\prod_{i=0}^{N-1} p(\Sigma_{t_{i+1}} | \Sigma_{t_{i}})$, observations $y^*_{0:N}$,  \\ 
$~~~~~~~~~~$A differential equation  $\frac{dx(t)}{dt} = f(x(t))$ with a numerical solver
%\STATE Randomly generate $K$ particles $ \{ \Sigma_{0|0:0}^k \}_{k=1}^K  \sim p(\Sigma_{t_0}|y_{t_0}^*)  $   //  Generate Particles at Time $t_0$
\STATE Generate $K$ particles $ \{ \Sigma_{0|0:0}^k \}_{k=1}^K  \sim p(\Sigma_{t_0}|y_{t_0}^*)  \propto p(\Sigma_{t_0}) p(y_{t_0}^* | \Sigma_{t_0})$  % //  Generate Particles at Time $t_0$
\FOR{$i = 0$ to $N-1$}
    \STATE $ \{\Sigma_{i+1 \mid 0:i}^k \}_{k=1}^K \sim p( \Sigma_{t_{i+1}} \mid \Sigma_{i \mid 0:i}^k )$ \quad $~~~~~~~~~~~~~~~~~~~~~~~~~~~~~~~~~~~~~~~~~~~~~~~~~~$ // Prediction Step
    \STATE Solve $\frac{dx(t)}{dt} = f(x(t))$ to  obtain  a numerical solution $x_{t_{i+1}}$.
    \STATE  For each $k$, evaluate  $p(y^*_{t_{i+1}}|\Sigma_{i+1 \mid 0:i}^k )$ with $x_{t_{i+1}}$, and obtain  \[ ~~~~~~~~~~~~~~~~~~~~~~~~~~~~~ w_k := \frac{p ( y^*_{t_{i+1}} |\Sigma_{i+1 \mid 0:i}^k  )}{\sum_{\tilde{k}=1}^K p ( y^*_{t_{i+1}}  \left| \right. \Sigma_{i+1 \mid 0:i}^{\tilde{k}}  )} \text{$~~~~~~~~~~~~~~~~~~~~~~~$ // Correction Step} \]    
    \STATE $\mathrm{re}_1, \ldots, \mathrm{re}_K \sim \text{Categorical}(w_1, \ldots, w_K)$.
\STATE 
Update particles by 
\[
~~~~~~~~~~~~~~~~~~~~~~~~~~~~~~~~~~~~~~~~~~ \Sigma_{i+1|0:i+1}^k := \Sigma_{i+1|0:i}^{\text{re}_k},
~~~~~~~~~~~~~~~~~~~~~~~~~~\text{// Resampling} \]
\[
~~~~~~~~~~~~~~~~~~~ 
\Sigma_{0 \mid 0:0}^k = \Sigma_{{0} \mid 0:0}^{\mathrm{re}_k}, \;\dots,\;
\Sigma_{j \mid 0:j}^k = \Sigma_{j \mid 0:j}^{\mathrm{re}_k}, \;\dots,\;
\Sigma_{i \mid 0:i}^k = \Sigma_{i \mid 0:i}^{\mathrm{re}_k}.
~~~~~~~~~\text{// Smoothing} \]
\ENDFOR
\STATE \textbf{Output:} Particles $\left\{ \left( \Sigma_{0|{0:0}}^k, \dots, \Sigma_{N|0:N}^k  \right) \right\}_{k=1}^K$  of $p\left(\Sigma_{0:N} \mid y^*_{0:N}\right)$.
\end{algorithmic}
\end{algorithm}
\hfill
\end{figure}

\subsubsection{Prediction Step}\label{subsubsec:prediction}
Assume that we have the particles $\{\Sigma_{i \mid 0:i}^k \}_{k=1}^K$ approximating the distribution $p\left( \Sigma_{t_i} \mid y_{0:i}^* \right)$.  
Then, the particles $\{\Sigma_{i+1 \mid 0:i}^k \}_{k=1}^K$ approximating $p\left( \Sigma_{t_{i+1}} \mid y_{0:i}^* \right)$ can be generated through
\[
\Sigma_{i+1 \mid 0:i}^k \sim p ( \Sigma_{t_{i+1}} \mid \Sigma_{i \mid 0:i}^k ),
\]
that is, by sampling from the Markov prior $p\left( \Sigma_{t_{i+1}} \mid \Sigma_{t_i} \right)$ with the conditioning $\Sigma_{t_i} = \Sigma_{i \mid 0:i}^k$.

\subsubsection{Correction Step}\label{subsubsec:correction}
Given new observation $y_{i+1}^*$ at the next time $t_{i+1}$, this step generates particles $\{ \Sigma_{i+1|0:i+1}^k \}_{k=1}^K$ that estimate
\begin{equation}
    p\left( \Sigma_{t_{i+1}} \mid y_{0:i}^*, y_{i+1}^* \right) = p\left( \Sigma_{t_{i+1}} \mid y_{0:i+1}^* \right)
    \propto p(y_{i+1}^* \mid \Sigma_{t_{i+1}} ) \cdot p( \Sigma_{t_{i+1}} \mid y_{0:i}^* ). \label{eq:correction_step}
\end{equation}
Recall that in the previous prediction step, we obtained a particle approximation
\( \{ \Sigma_{i+1|0:i}^k \}_{k=1}^K \) of the distribution
\( p( \Sigma_{t_{i+1}} \mid y_{0:i}^* ) \).
Therefore, the distribution $p\left( \Sigma_{t_{i+1}} \mid y_{0:i+1}^* \right)$ can be approximated as
\begin{equation*}
p\left( \Sigma_{t_{i+1}} \mid y_{0:i+1}^* \right)
\propto p(y_{i+1}^* \mid \Sigma_{t_{i+1}} ) \cdot p( \Sigma_{t_{i+1}} \mid y_{0:i}^* )
\approx \sum_{k=1}^K w_k \, \delta_{ \Sigma_{i+1|0:i}^k },
\end{equation*}
where the weights $w_k$ are defined as
\begin{equation}\label{eq:particle_weight}
w_k := \frac{p ( y_{i+1}^* \mid \Sigma_{i+1|0:i}^k )}
{\sum_{\tilde{k}=1}^K p ( y_{i+1}^* \mid \Sigma_{i+1|0:i}^{\tilde{k}} )}.
\end{equation}

We then perform \emph{resampling} to equalize the particle weights $w_k$ \cite{kitagawa1996monte, gordon1993novel}.
Specifically, $K$ indices $\mathrm{re}_1, \dots, \mathrm{re}_K$ ($\in \left\{1, \dots, K \right\}$) are sampled from 
$\categorical(w_1, \ldots, w_K)$, i.e., the distribution in which index $k \in \left\{1, \dots , K \right\}$ is selected with probability $w_k$.
Then, the particles are subsequently updated as
\begin{equation}\label{resmpling}
\Sigma_{i+1|0:i+1}^k := \Sigma_{i+1|0:i}^{\text{re}_k}.
\end{equation}
This procedure yields a new particle set $\{\Sigma_{i+1|0:i+1}^k \}_{k=1}^K$ with uniform weights, 
which approximates the posterior $p\left( \Sigma_{t_{i+1}} \mid y_{0:i+1}^* \right) $.

\begin{remark}
    It is well known that the resampling procedure gradually reduces particle diversity over time: this phenomenon is commonly referred to as \emph{particle degeneracy}.
    This issue becomes particularly severe in high-dimensional latent spaces, where the posterior distribution must be represented by a limited number of particles.
    Various strategies have been proposed to mitigate this problem.
    One such method is the \emph{merging particle filter} \cite{nakano2007merging,van2009particle}, in which multiple particles are generated and then merged to form a single representative particle.
    Another approach is the \emph{localized particle filter} \cite{farchi2018comparison,van2019particle}, based on the assumption that observations depend only on a subset of the latent state.
    Additional variants include the \emph{tempering particle filter} \cite{beskos2014stability}, which employs tempering techniques \cite{neal1996sampling,del2006sequential}, 
    and the \emph{implicit particle filter} 
    \cite{atkins2013implicit,chorin2010implicit},
    which alleviates degeneracy through the design of appropriate proposal distributions.
    Since the primary objective of the present paper is to illustrate an approach for inferring ODE parameters with discretization error quantification within the framework of state-space modelling, rather than to propose a new filtering algorithm, the standard particle filter is employed; nonetheless, the enhanced variants mentioned above could also be utilized.
\end{remark}

\subsubsection{Smoothing}\label{subsubsec:smoothing}
In the operations described so far, the particles $\{\Sigma_{i \mid 0:i}^k \}_{k=1}^K$ obtained at times $t_i = t_0, \dots, t_{N-1}$ are sampled from 
$p\!\left( \Sigma_{t_i} \mid y_{0:i}^* \right)$ rather than from the distribution that incorporates all future observations, 
$p\!\left( \Sigma_{t_i} \mid y_{0:N}^* \right)$. 
This limitation is addressed by a procedure known as \emph{smoothing}.

In a particle filter, a principled way to perform smoothing is to update
\begin{equation}\label{eq_smoothing}
\Sigma_{0 \mid 0:0}^k := \Sigma_{0 \mid 0:0}^{\mathrm{re}_k}, \;\dots,\;
\Sigma_{j \mid 0:j}^k := \Sigma_{j \mid 0:j}^{\mathrm{re}_k}, \;\dots,\;
\Sigma_{i \mid 0:i}^k := \Sigma_{i \mid 0:i}^{\mathrm{re}_k},
\end{equation}
in the resampling step~\eqref{resmpling}, together with 
$\Sigma_{i+1 \mid 0:i+1}^k := \Sigma_{i+1 \mid 0:i}^{\mathrm{re}_k}$. 
Through this operation, the particles $\Sigma_{j \mid 0:j}^k$ ($0 \leq j \leq i$) become samples from 
$p\!\left( \Sigma_{t_j} \mid y_{0:i+1}^* \right)$, which incorporates the observations $y^*_{t_{0}}, \dots , y^*_{t_{i+1}}$ beyond time $j$~\cite{kitagawa1996monte} (see Appendix~\ref{appendix_smoothing} for details).  
By repeating this procedure until $i = N - 1$, the resulting particles $\Sigma_{j \mid 0:j}^k$ approximate draws from 
$p\!\left( \Sigma_{t_j} \mid y^*_{0:N} \right)$.

\begin{remark}
The resampling procedure \eqref{eq_smoothing} simultaneously updates the weight indices for all time steps, thereby reducing the particle diversity even at early stages ($j \ll N$). 
One way to mitigate this issue is \emph{fixed-lag smoothing}~\cite{kitagawa1996monte}, in which resampling is restricted to only the past $L$ steps:
\begin{equation}
\Sigma_{i-L \mid 0:i-L}^k = \Sigma_{i-L \mid 0:i-L}^{\mathrm{re}_k}  ,\dots, 
\Sigma_{i \mid 0:i}^k = \Sigma_{i \mid 0:i}^{\mathrm{re}_k}.
\end{equation}
This approach confines particle degeneration due to resampling to the most recent $L$ time steps.
In addition, the various resampling techniques introduced in the previous section to mitigate particle degeneration can also be applied here to further suppress their effects during smoothing.
\end{remark}

\subsection{Joint Bayesian Inference of ODE Parameters and Discretization Error Variances}\label{subsec:JointBayes}

We return to our objective \eqref{eq:our_goal}---obtaining the posterior
$
p\left( \theta, \Sigma_{t_0}, \dots, \Sigma_{t_N} \mid y_{t_0}, \dots, y_{t_N} \right)
$
over both the discretization error variances $\Sigma_{t_i}$ and the ODE parameter $\theta$.

We now reformulate our target posterior distribution \eqref{eq:our_goal} as a state-space modeling problem, based on the discussion in the previous section. 
First, we introduce a conditional Markov prior 
$
p(\theta, \Sigma_{t_0}, \dots, \Sigma_{t_N})
$
over $(\theta, \Sigma_{t_0}, \dots, \Sigma_{t_N})$, in which the sequence $(\Sigma_{t_0}, \dots, \Sigma_{t_N})$ satisfies the Markov property conditioned on the model parameter $\theta$:
\[
p(\theta, \Sigma_{t_0}, \dots, \Sigma_{t_N}) 
= p(\theta) \cdot p(\Sigma_{t_0} \mid \theta) \,
\prod_{i=0}^{N-1} p(\Sigma_{t_{i+1}} \mid \Sigma_{t_i}, \theta).
\]
As the likelihood $p(y_{t_0}, \ldots, y_{t_N} \mid \theta, \Sigma_{t_0}, \ldots, \Sigma_{t_N})$ 
can be factorized as 
$\prod_{i=0}^N p(y_{t_i} \mid \theta, \Sigma_{t_i})$
in \eqref{eq:likelihood_with_discretization}, 
the joint model $p(y_{t_0}, \ldots, y_{t_N}, \theta, \Sigma_{t_0}, \ldots, \Sigma_{t_N})$ 
can be expressed as
\begin{align*}
   p(y_{t_0},..., y_{t_N}, \theta, \Sigma_{t_0},..., \Sigma_{t_N})  &=  p\left( y_{t_0}, \dots, y_{t_N} \mid \theta, \Sigma_{t_0}, \dots, \Sigma_{t_N} \right) \cdot p(\theta, \Sigma_{t_0}, \dots, \Sigma_{t_N}) \\
   & = \left( \prod_{i=0}^N p(y_{t_i} \mid \Sigma_{t_i},\theta )\right)\left( p(\theta) \cdot {p}(\Sigma_{t_0} \mid \theta)\, 
\prod_{i=0}^{N-1} p(\Sigma_{t_{i+1}} \mid \Sigma_{t_i}, \theta) \right)  \\
& = \underbrace{p(\theta)}_{(*)} \cdot \underbrace{%
\Bigl( p(\Sigma_{t_0} \mid \theta) \cdot \prod_{i=0}^N p(y_{t_i} \mid \Sigma_{t_i},\theta) \cdot  
\prod_{i=0}^{N-1} p(\Sigma_{t_{i+1}} \mid \Sigma_{t_i}, \theta) \Bigr). 
}_{(**)}
\end{align*}
Thus, the joint model 
$
p(y_{t_0}, \ldots, y_{t_N}, \theta, \Sigma_{t_0}, \ldots, \Sigma_{t_N})
$ 
consists of a state-space model $(**)$, 
where both the latent state transitions $p(\Sigma_{t_{i+1}} \mid \Sigma_{t_i}, \theta)$
and the observation model $p(y_{t_i} \mid \Sigma_{t_i}, \theta)$ depend on the model parameter $\theta$, 
together with a prior distribution $p(\theta)~(*)$ (see Figure~\ref{fig:state_space}).  
Here, the observation process 
$p(y_{t_i} \mid \Sigma_{t_i}, \theta)$ is given by
\begin{align}
    p(y_{t_i} \mid {\Sigma}_{t_i}, \theta) &:= \mathcal{N}(H x_{t_i}(\theta), \Gamma + V_{t_i}) \notag \\
    &= \frac{1}{\sqrt{(2\pi)^{d_{\mathcal{Y}}} \lvert \Gamma + V_{t_i} \rvert}}
    \exp\Bigg\{
        -\frac{1}{2} \left( y_{t_i} - H x_{t_i}(\theta) \right)^{\mathsf{T}}
        (\Gamma + V_{t_i})^{-1}
        \left( y_{t_i} - H x_{t_i}(\theta) \right)
    \Bigg\}, 
    \label{eq:ml_ySigma_self}
\end{align}
which is a modified version of the observation process \eqref{eq:ml_ySigma} and compatible to the modified likelihood \eqref{eq:likelihood_with_discretization}
where the numerical solution $x_{t_i}(\theta)$ explicitly depends on the model parameter $\theta$. 
This reformulation implies that 
inferring the posterior
$
p\left( \theta, \Sigma_{t_0}, \dots, \Sigma_{t_N} \mid y_{t_0}, \dots, y_{t_N} \right)$
in \eqref{eq:our_goal}
can be cast as a joint Bayesian inference problem for the latent states $\Sigma_{t_i}$ and the model parameter $\theta$ in the state-space model $(**)$.

\begin{figure}[t]
    \centering

    \begin{tikzpicture}
    \node[circle, draw, text width=1.4cm,align=center] (c1) {$\begin{pmatrix}
\Sigma_{t_{i-1}} \\
\theta_{t_{i-1}}
\end{pmatrix}$};
    \node[circle, draw, text width=1.4cm,align=center, right = 3.6cm of c1] (c2) {$\begin{pmatrix}
\Sigma_{t_i} \\
\theta_{t_i}
\end{pmatrix}$};
    \node[circle, draw, text width=1.3cm,align=center, above = 1.4cm of c1] (c3) {$y_{t_{i-1}}$};
    \node[circle, draw, text width=1.3cm,align=center, above = 1.4cm of c2] (c4) {$y_{t_i}$};
    \node[left = 3.2cm of c1] (c5) {};
    \node[right = 3.2cm of c2] (c6) {};

\draw[-{Stealth[length=2.5mm]}, thick]
  (c1)
  --
  node[above, sloped, yshift=3pt, fill=white, inner sep=1pt]
       {\small{$\Sigma_{t_i} \sim p(\Sigma_{t_i}\mid \Sigma_{t_{i-1}},\theta)$}}
     node[below, sloped, yshift=-3pt, fill=white, inner sep=1pt]
       {\small{$\theta_{t_i} = \theta_{t_{i-1}}$}}
  (c2);
    \draw [-{Stealth[length=2mm]}, thick] (c1)--node[fill=white,sloped,rotate=-90]{$\mathcal{N}(H x_{t_i-1} (\theta_{t_i-1}),\Gamma + V_{t_{i-1}})$} (c3);
    \draw [-{Stealth[length=2mm]}, thick] (c2)--node[fill=white,sloped,rotate=-90]{$\mathcal{N}(H x_{t_i} (\theta_{t_i}),\Gamma + V_{t_i})$} (c4);
\draw[-{Stealth[length=2.5mm]}, thick]
  (c5)
  -- node[above, sloped, yshift=4pt, fill=white, inner sep=1pt,  xshift=-14pt]
       {\small{$\Sigma_{t_{i-1}} \sim p(\Sigma_{t_{i-1}}\mid \Sigma_{t_{i-2}}, \theta)$}}
     node[below, sloped, yshift=-5pt, fill=white, inner sep=1pt,  xshift=-14pt]
       {\small{$\theta_{t_{i-1}} = \theta_{t_{i-2}}$}}
  (c1);
\draw[-{Stealth[length=2.5mm]}, thick]
  (c2)
  -- node[above, sloped, yshift=4pt,   xshift=9pt, fill=white, inner sep=1pt]
       {\small{$\Sigma_{t_{i+1}} \sim p(\Sigma_{t_{i+1}}\mid \Sigma_{t_{i}}, \theta)$}}
     node[below, sloped, yshift=-4pt, xshift=9pt, fill=white, inner sep=1pt]
       {\small{$\theta_{t_{i+1}} = \theta_{t_{i}}$}}
  (c6);
    \end{tikzpicture}
    
    \caption{Self-organizing state-space modeling to perform joint Bayesian inference of the latent state $\Sigma_{t_i}$ and the ODE parameters $\theta$ shown in Figure~\ref{fig:state_space}. In this method, we construct an augmented state-space model by considering the joint latent state $(\Sigma_{t_i}, \theta_{t_i})^{\mathrm{T}}$, and define its transition as $\Sigma_{t_{i+1}} \sim p(\Sigma_{t_{i+1}} \mid \Sigma_{t_i}, \theta_{t_i})$ and $\theta_{t_{i+1}} = \theta_{t_i}. $ By solving this alternative state-space model using a particle filter, we obtain the desired posterior distribution (\ref{eq:our_goal}).
}
    \label{fig:self_state_space}
\end{figure}
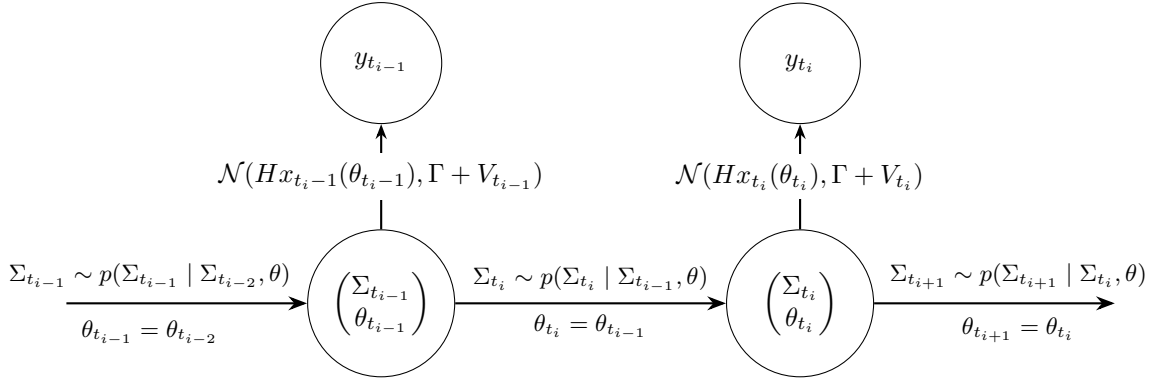

\begin{figure}[tp]
\centering
\begin{algorithm}[H]
\caption{Particle Filter for Bayesian Inference on $\Sigma_{t_i}$ and $\theta$ }\label{alg:particle_self}
\begin{algorithmic}[1]
\STATE \textbf{Input:} A prior $p(\theta, \Sigma_{0:N}) = p(\theta) \cdot p(\Sigma_{t_{0}} | \theta)\prod_{i=0}^{N-1} p(\Sigma_{t_{i+1}} | \Sigma_{t_{i}}, \theta)$, observations $y^*_{0:N}$, \\
$~~~~~~~~~~$  A differential equation  $\frac{dx(t)}{dt} = f(x(t), \theta)$ with a numerical solver 
\STATE Generate particles $ \left\{ \left( \Sigma_{0|0:0}^k, \theta_{0|0:0}^{k} \right) \right\}_{k=1}^K \sim 
%p(\Sigma_{t_0}|\theta) \cdot p(\theta) $ //  
p(\Sigma_{t_0}, \theta| y_{t_0}^*) \propto p( y_{t_0^* } | \theta,\Sigma_{t_{0}} ) \cdot p(\theta) \cdot p(\Sigma_{t_{0}} | \theta)$ 
%//Generate Particles at Time $t_0$
\FOR{$i = 0$ to $N-1$}
    \STATE Generate particles $ \left\{ \left( \Sigma_{i+1|0:i}^k, \theta_{i+1|0:i}^{k} \right) \right\}_{k=1}^K$ by   \\ $~~~~~~~~~~~~~~~\Sigma_{i+1|0:i}^{k} \sim p(\Sigma_{t_{i+1} } \mid \Sigma_{i|0:i}^k, {\theta_{i+1|0:i}^{k}}  ), ~~~ \theta_{i+1|0:i}^{k} :=\theta_{i|0:i}^{k} $ \quad $~~~~~$ // Prediction Step
    \STATE For each  $k$, solve $\frac{dx(t)}{dt} = f(x(t), \theta_{i+1|0:i}^{k})$ to  obtain a numerical solution $x_{t_{i+1}}(\theta_{i+1|0:i}^{k})$.
    \STATE  For each  $k$, evaluate  $p(y^*_{i+1}|\Sigma_{i+1|0:i}^k, \theta_{i+1|0:i}^{k})$ with $x_{t_{i+1}}(\theta_{i+1|0:i}^{k})$ and obtain \[ ~~~~~~~~~~~~~~~~~~~~~~~~~~~~~ w_k := \frac{p ( y^*_{i+1} |\Sigma_{i+1 \mid 0:i}^k, ~\theta_{i+1 \mid 0:i}^k  )}{\sum_{\tilde{k}=1}^K p ( y^*_{i+1} \left| \right. \Sigma_{i+1 \mid 0:i}^{\tilde{k}} , ~\theta_{i+1 \mid 0:i}^{\tilde{k}} )} \text{$~~~~~~~~~~~~$ // Correction Step} \]

    \STATE $\mathrm{re}_1, \ldots, \mathrm{re}_K \sim \text{Categorical}(w_1, \ldots, w_K)$.
$~~~~~~~~~~~~~~~~~~~~~~~~~~~~~~~~~~~~~~~$ 
\STATE 
Update particles by 
\[ ~~~~~~~~~~~~~~~~~~~~~~~~ \Sigma_{i+1|0:i+1}^{k} := \Sigma_{i+1|0:i}^{\text{re}_k}, ~~
\theta_{i+1|0:i+1}^{k} := \theta_{i+1|0:i}^{\text{re}_k} 
~~~~~~~~~~~~~~~~ \text{// Resampling} \]
\[
~~~~~~~~~~~~~~~~~~~ 
\Sigma_{0 \mid 0:0}^k = \Sigma_{{0} \mid 0:{0}}^{\mathrm{re}_k}, \;\dots,\;
\Sigma_{j \mid 0:j}^k = \Sigma_{j \mid 0:j}^{\mathrm{re}_k}, \;\dots,\;
\Sigma_{i \mid 0:i}^k = \Sigma_{i \mid 0:i}^{\mathrm{re}_k}.
~~~~~~~~~\text{// Smoothing} \]
\ENDFOR
\STATE \textbf{Output:} Particles $\left\{ \left( \theta_{N|0:N}^k, \Sigma_{0|0{:0}}^k, \dots, \Sigma_{N|0:N}^k  \right) \right\}_{k=1}^K$  of $p\left(\theta, \Sigma_{0:N} \mid y^*_{0:N} \right)$.
\end{algorithmic}
\end{algorithm}
\hfill
\end{figure}

Since we have reformulated our objective \eqref{eq:our_goal} as a joint Bayesian inference for the latent states and model parameters within a state-space model, we can exploit existing methods developed for such joint inference in the context of the filtering approach.  
A well-known method of this kind is the \emph{self-organizing state-space model}~\cite{kitagawa1998self-organizing}, which is summarized in Algorithm \ref{alg:particle_self}.
In this method, an alternative state-space formulation, as illustrated in Figure~\ref{fig:self_state_space}, is considered, in which the model parameter $\theta$ is incorporated into the latent state $\Sigma_{t_i}$ by introducing the augmented latent variable 
\begin{equation*}
    \hat{\Sigma}_{t_i} =
\begin{pmatrix}
    \Sigma_{t_i} \\
    \theta_{t_i}
\end{pmatrix},
\end{equation*}
where the parameters $\theta$ are formally treated as time-varying variables $\theta_{t_i}$ with constant transition $\theta_{t_{i+1}} = \theta_{t_i}$, even though they are originally time-invariant.
We define a Markov distribution over $\hat{\Sigma}_{t_i}$ by initializing 
\[
\hat{\Sigma}_{t_0} = (\Sigma_{t_0}, \theta_{t_0})^\top \sim p(\Sigma_{t_0} \mid \theta) \cdot p(\theta)
\] 
at the initial time $t_0$, and specifying the transition distribution $p(\hat{\Sigma}_{t_{i+1}} \mid \hat{\Sigma}_{t_i})$ as
\[
\Sigma_{t_{i+1}} \sim p(\Sigma_{t_{i+1}} \mid \Sigma_{t_i}, \theta_{t_i}), \quad
\theta_{t_{i+1}} = \theta_{t_i}.
\]
For the augmented latent space $\hat{\Sigma}_{t_i} = (\Sigma_{t_i}, \theta_{t_i})^\top$, we define the observation process as 
\[
%p(y_{t_i} \mid \hat{\Sigma}_{t_i}) = 
p(y_{t_i} \mid \Sigma_{t_i}, \theta_{t_i}) = \mathcal{N}\big(H x_{t_i}(\theta_{t_i}), \Gamma + V_{t_i}\big),
\] 
as given in \eqref{eq:ml_ySigma_self}. By applying a particle filter to the augmented state-space model illustrated in Figure~\ref{fig:self_state_space}, 
we can sample from the posterior distribution 
$
p(\theta_{t_N}, \Sigma_{t_0}, \ldots, \Sigma_{t_N} \mid y_{t_0}, \ldots, y_{t_N}).
$ 
Since $\theta_{t_i}$ follows a constant transition, the sequence $(\theta_{t_N}, \Sigma_{t_0}, \ldots, \Sigma_{t_N})$ of latent variables in the augmented state-space model (Figure~\ref{fig:self_state_space}) follows the original prior distribution $p(\theta, \Sigma_{t_0}, \ldots, \Sigma_{t_N})$:
\begin{align*}
(\theta_{t_N}, \Sigma_{t_0}, \ldots, \Sigma_{t_N}) \sim p(\theta_{t_N}, \Sigma_{t_0}, \ldots, \Sigma_{t_N}) = p(\theta_{t_0}, \Sigma_{t_0}, \ldots, \Sigma_{t_N}) = p(\theta, \Sigma_{t_0}, \ldots, \Sigma_{t_N}).
\end{align*} 
Correspondingly, the resulting posterior distribution 
$
p(\theta_{t_N}, \Sigma_{t_0}, \ldots, \Sigma_{t_N} \mid y_{t_0}, \ldots, y_{t_N})
$ 
is identical to our target posterior 
$
p(\theta, \Sigma_{t_0}, \ldots, \Sigma_{t_N} \mid y_{t_0}, \ldots, y_{t_N}).
$

\section{Markov Prior on Discretization Error Variances}\label{sec:prior}

It remains to construct an appropriate Markov prior for the discretization error variances.
There are various possible constructions depending on the problem setting, such as the type of ODEs, the numerical method employed, and the available observations.
In this section, we review the local propagation of the global error and, based on this analysis, propose a suitable prior.

\subsection{Local Propagation of the Global Error}

Here, we briefly review the standard theory of error analysis \cite{butcher2016numerical, hairer2020solving}.
We distinguish between the \emph{local error} and the \emph{global error}, and discuss the propagation of the latter.

In this subsection, we omit $\theta$ for simplicity.
We also make the following assumptions:
\begin{itemize}
\item The observation time interval is constant, i.e. $t_{i+1} - t_i = \text{const.}$
\item The step size used in the time integrator is denoted by $h$, and $t_{i+1} - t_i = kh$ for a positive integer $k$.
\end{itemize}
The first assumption is introduced solely for clarity of presentation.
If the second assumption does not hold in practical applications, the solution at $t = t_i$ can be approximated, for example, by interpolation techniques using neighboring numerical solutions~\cite{hairer2020solving}.

The numerical solution at {$t_{i,j}=t_i+jh$} is denoted by $x_{i,j}$ so that $x_{i+1} = x_{i,k}$.
The time-$h$ flow for the ODE is denoted by $\phi_h$,
and the time-$h$ flow of the numerical solver 
by $\psi_h$.
The local error is defined as the numerical error induced in a single time step, i.e. $\psi_h(x) - \phi_h (x)$.
We define 
\begin{equation*}
    L(t_{i,j}) := \psi_h (x_{i,j}) - \phi_h(x_{i,j}) 
    = x_{i,j+1} - \phi_h(x_{i,j}) .
\end{equation*}

On the other hand, the global error is defined by
\begin{equation*}
    G(t_{i,j}) = x_{i,j} - x(t_i+jh).
\end{equation*}
We now examine how the global error propagates.
Observe that
\begin{align*}
    G(t_{i,j+1})
    & = x_{i,j+1} - x(t_i+(j+1)h)
    = \phi_h(x_{i,j}) - \phi_h(x(t_i+jh)) + x_{i,j+1} - \phi_h(x_{i,j})  \\
    & = \phi_h(x_{i,j}) - \phi_h(x(t_i+jh)) + L(t_{i,j}).
\end{align*}
To analyze the difference $\phi_h(x_{i,j}) - \phi_h(x(t_i+jh))$,
we consider the variational equation
\begin{equation*}
    \frac{d\delta(t)}{dt} = \big(\nabla_x f(x(t))\big) \delta (t), \quad 
    \delta(0) = x_{i,j} - x(t_i+jh),
\end{equation*}
{where $\nabla_x f(x(t)) \in \mathbb{R}^{d_{\mathcal X} \times d_{\mathcal X}}$ 
denotes the Jacobian matrix of $f:\mathbb{R}^{d_{\mathcal X}} \rightarrow \mathbb{R}^{d_{\mathcal X}}$.}
For the solution to this equation, it follows that
\begin{equation*}
    \phi_h(x_{i,j}) - \phi_h(x(t_i+jh))
     = \delta(h) + {o} (\|  x_{i,j} - x(t_i+jh)  \|).
\end{equation*}
By expanding $\delta(h)$ in a Taylor series, we obtain an approximation
\begin{equation*}
    \delta(h) 
    = \delta (0) + h \big( \nabla_x f(x_{i,j})\big) \delta(0) + \mathcal{O}(h^2) \delta (0).
\end{equation*}
Hence, the propagation of the global error satisfies
\begin{equation}
    \label{eq:ge_propagation1}
    G(t_{i,j+1})
    \approx  (I+ h \nabla_x f(x_{i,j}) + \mathcal{O}(h^2) ) G(t_{i,j}) + L(t_{i,j}).
\end{equation}

\subsection{A Markov Prior}\label{subsec:markov_prior}

Constructing a suitable Markov prior for discretization error variances is crucial.  
Note that the prior introduced in the previous study~\cite{miyatake2025quantifying} does not satisfy the Markov property and is therefore not appropriate for our setting.
In this section, we propose a new Markov prior, motivated by the preceding subsection.

Throughout, we assume that the discretization error variance is diagonal and takes the form
\begin{equation}\label{eq:assumption_dist_error}
\Sigma_{t_{i}} = \mathrm{diag}\Big( \left( \sigma_{t_{i}}^1\right)^2,\dots,\big( \sigma_{t_{i}}^{d_\X}\big)^2 \Big),
\end{equation}
and define the vector
\begin{equation}\label{eq:discretization_vector}
\mathbf{\sigma}_{t_{i}} := \bigl(\sigma_{t_{i}}^1,\dots,\sigma_{t_{i}}^{d_\X}\bigr)^\top \in \mathbb{R}_{\geq 0}^{d_\X}.
\end{equation}
We also use the notation $\Sigma_{t_{i,j}}$ and $\sigma_{t_{i,j}}$ with their obvious meanings.

Note that the discretization error {\em standard deviations} $\sigma_{t_{i,j}}$ in our context can be interpreted as a model for the global error $G(t_{i,j})$.  
By identifying $G(t_{i,j})$ in \eqref{eq:ge_propagation1} with $\sigma_{t_{i,j}}$, we arrive at the following Markov prior, in which  
the discretization error standard deviation $\sigma_{t_{i,j}}$ evolves according to the probabilistic transition below as time advances by $h$:
\begin{equation}\label{eq:proposed_prior}
    \sigma_{t_{i,j+1}} = M_{i, j} \cdot \sigma_{t_{i,j}} + |\tilde{L}(t_{i,j})|,
    \qquad M_{i, j} \sim P \quad 
    (j = 0, \dots, k-1).
\end{equation}
Here, $P$ denotes a distribution over matrices $M_{i, j} \in \mathbb{R}_{>0}^{d_\mathcal{X} \times d_\mathcal{X}}$. 
Each component of $M_{i, j}$ is constrained to be nonnegative.
The value of $\tilde{L}(t_{i,j})$ is an approximated value of $L(t_{i,j})$.
This can be typically estimated by, for example, either $\tilde{L}(t_{i,j}) = \psi_h (x_{i,j}) - (\psi_{h/2} \circ \psi_{h/2}) (x_{i,j})$
or 
$\tilde{L}(t_{i,j}) = \psi_h (x_{i,j}) - \tilde{\psi}_h (x_{i,j})$, where $\tilde{\psi}_h$ denotes a higher order numerical solver.
These techniques are commonly used to control the step size during time integration~\cite{hairer2020solving}.
The absolute value $|\cdot|$ of $L(t_{i,j})$ is taken componentwise to ensure consistency with \eqref{eq:discretization_vector}, which was defined to be nonnegative.

The proposed Markov prior involves a distribution $M \sim P_\lambda$, which typically depends on a hyperparameter $\lambda$.  
In our experiment, for instance, we set $P$ as $m \cdot I$, where $m$ is drawn from a Gamma distribution $m \sim {\text{Gam}} (\alpha, \beta)$; in this case, the hyperparameter $\lambda = (\alpha, \beta)$ must be properly chosen. {In this study, the hyperparameter is tuned via an empirical Bayes approach by maximizing the log marginal likelihood $\log p(y^*(t_{0:N}) \mid \lambda)$ over candidates $\lambda$. In particle filters, this likelihood is computed using the normalizing constants of the weights in \eqref{eq:particle_weight}; thus, a nearly optimal hyperparameter is obtained by iterating lines 3–9 of Algorithm \ref{alg:particle_self} for each candidate.
}

As shown in the next section, the empirically optimal parameters often satisfy the mean condition
$
    \mathbb{E}_{M \sim P_\lambda}[ M ] \approx I,
$
where $I$ denotes the identity matrix. 
This observation suggests that it suffices to restrict attention to hyperparameter candidates $\lambda$ such that
$
\mathbb{E}_{M \sim P_\lambda}[ M ] = I.
$
In the next subsection, we examine this empirical finding from a theoretical perspective.

\begin{remark}\label{remark:m_expectation}
    We have introduced the matrix $M_{i, j}$ in \eqref{eq:proposed_prior} expecting that it plays a similar role to the term 
    $(I+ h \nabla_x f(x_{i,j}) + \mathcal{O}(h^2))$ 
    in \eqref{eq:ge_propagation1}.
    In this viewpoint,
    assuming the non-negativity does not pose a problem.
    Moreover, under some smoothness conditions on $f$, it follows that 
    $(I+ h \nabla_x f(x_{i,j}) + \mathcal{O}(h^2)) \to I$ 
    as $h\to +0$, which supports empirical finding mentioned above.
    
    However, taking the absolute value $|\tilde{L}(t_{i,j})|$ of the estimated local error may lead to overestimation.
    In view of this issue, there is potential to develop a more appropriate model or prior.
    Nonetheless, we emphasize that the current approach remains effective in many situations, due to the following:
    \begin{itemize}
        \item as shown in the next subsection, the prior satisfies a desirable asymptotic property;
        \item in phases where the global error increases, adding a positive value appears natural;
        \item while there might be phases during which the absolute value of the global error decreases (when evaluated componentwise), such phenomena can be captured in the corresponding posterior distribution.
    \end{itemize}
\end{remark}

\subsection{Asymptotics of the Prior}

We investigate asymptotic properties of the proposed prior \eqref{eq:proposed_prior} as the step size $h$ approaches zero. 
For simplicity in the theoretical analysis, we identify $\tilde{L}(t_{i,j})$ with the exact local error $L(t_{i,j})$, although in practice we estimate it using two numerical solvers. To explicitly show the step-size dependence of the proposed prior,
we write the prior as
\begin{equation}\label{eq:prior_stepsize}
    p_h \left( \sigma_{t_{0, 0}}, \ldots, \sigma_{t_{N, 0}} \right)= p_h  \left( \sigma_{t_0}, \ldots, \sigma_{t_N})
    = p_h (\sigma_{t_0} \right) \cdot \prod_{i=0}^{N-1} p_h (\sigma_{t_{i+1}} \mid \sigma_{t_i}).
\end{equation}

Numerical solvers are typically designed so that the numerical error vanishes asymptotically as the step size \( h \to + 0 \). 
Accordingly, the proposed prior \( p_h(\sigma_{t_i}) \) should be defined so that \( \sigma_{t_i} \to 0 \) as \( h \to + 0 \) in the sense of probabilistic convergence. This asymptotic property can be guaranteed by the following theorem, which also provides the corresponding convergence rate:

\begin{theo}[Convergence rate in the step-size limit]\label{theo_convergence_rate} Assume that there exist constants $C_M, C_L, C_\sigma \in \bR_{>0}$  satisfying the following conditions: 
\begin{itemize}
    \item[(I)] A distribution $P_{\lambda, h}$ on matrices $M_{i,j}$, or equivalently $P_{\lambda}$ with $\lambda = \lambda(h)$,  satisfies
$
\mathbb{E}_{M_{i,j}\sim P_\lambda}\!\left[\|I-M_{i,j}\|_{\mathrm{op}}\right]
\leq C_M h,
$
where $\|\cdot\|_{\mathrm{op}}$ denotes the operator norm of matrices.
    \item[(II)]$\| L(t_{i,j}) \| \leq C_L  h^{a+1}$ for \( 0 \leq i \leq N-1 \) and  \(  0 \leq j \leq k-1 \).

    \item[(III)] A prior $p_h(\sigma_{t_0})$ on initial time $t_0$ satisfies $\mathbb{E}\!\left[\|\sigma_{t_0}\| \right]
\leq C_\sigma h^{b}$ for  $b \geq 0$.
    \item[(IV)] $M_{i, j} \sim P_\lambda$ is independent of $\sigma_{t_{i, j}}$.
\end{itemize}
Then, for $i \geq 1$, we have
\begin{equation}
   \sigma_{t_i}   = \mathcal{O}_p (h^{\min({a, b})})~~~~~~~(h\rightarrow +0).
\end{equation}

\end{theo}

{The proof is given in Appendix \ref{appendix_proof_asymot}.} Condition (I) requires that the distribution $P_\lambda$ be constructed such that 
$
\mathbb{E}_{M_{i, j} \sim P_\lambda}\big[\| I - M_{i,j} \|_{\text{op}}\big] = \mathcal{O}(h).
$
%This condition can be easily satisfied by defining a distribution $P$ for which an analytical expression of  $\mathbb{E}_{M_{i,j} \sim P}\big[\| I - M_{i,j} \|_{\text{op}}\big]$ is available, including the one used in our experiments. 
This condition implies that \(P_\lambda\) is, on average, concentrated around the identity matrix \(I\),
which is consistent with Remark~\ref{remark:m_expectation}, asserting that
\(P_\lambda\) converges to the Dirac measure at \(I\) as \(h \to +0\).
%and converge to a Dirac measure at \(I\) as \(h \to +0 \), which is consistent with Remark~\ref{remark:m_expectation}.
Condition (II) defines a given numerical method to be of order ${a}$.
Condition (III) indicates that the choice of \({b}\) in the proposed prior at \(t_0\) affects the convergence rate: the prior achieves \(\mathcal{O}_p(h^{{a}})\) if \({b} \geq {a}\), and \(\mathcal{O}_p(h^{{b}})\) otherwise.

This theorem shows that the proposed prior is likely to converge in probability at the same rate as the global error $G(t_{i, j}) = x_{t_{i, j}} - x(t_{i, j})$. It is well known that, for many numerical methods, if the local error $L(t_{i, j})= \mathcal{O} (h^{{a}+1})$, then the global error $G(t_{i, j})$ is $\mathcal{O}(h^{{a}})$~\cite{butcher2016numerical, hairer2020solving, iserles2009first}. Since the discretization error standard deviation in this study corresponds to the global error, the proposed prior is desired to be $\mathcal{O}_p (h^{{a}})$, in order to incorporate this numerical insight into the prior. This theorem shows that the objective is achieved by designing~$P$ to satisfy condition~(I) and by setting the initial distribution~$p_h(\sigma_{t_0})$ such that $\mathbb{E}_{p_h(\sigma_{t_0})}[\| \sigma_{t_0}\|] = \mathcal{O}{(h^{{b}})}$ with $ {b} \geq {a}$.

\section{Experiments}
\label{sec:experiments}

We evaluate the performance of the proposed method using {three} examples: {the chemical reaction network,} the pendulum system and the FitzHugh--Nagumo model.
These are representative cases where coarse numerical integration, combined with parameter values that deviate significantly from the true ones, can coincidentally produce trajectories that closely match the observations or the exact solution (obtained using the true parameters).

For all cases, we employ the explicit Euler method as our numerical solver, and also use Runge's method (the explicit trapezoidal rule) to estimate the local discretization errors.
The explicit Euler method is chosen to represent a scenario in which discretization errors may be significant.
The observations are generated by adding randomly sampled noise to the reference solution, which is computed using a sufficiently accurate solver.
The reference solution is also used to calculate the \emph{exact} discretization errors.
All numerical experiments were implemented in \texttt{Julia}.

\begin{figure}[t]
  \centering

  \begin{subfigure}{0.48\textwidth}
    \centering
    \includegraphics[width=\linewidth]{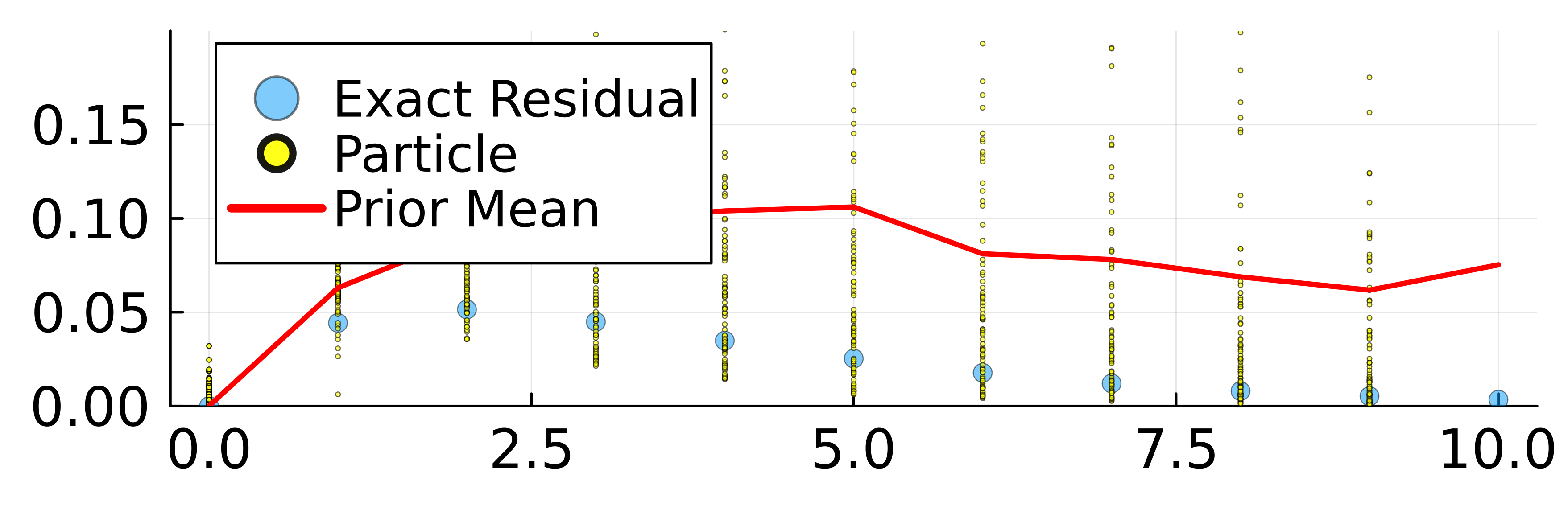}
  \end{subfigure}
  \hfill
  \begin{subfigure}{0.48\textwidth}
    \centering
    \includegraphics[width=\linewidth]{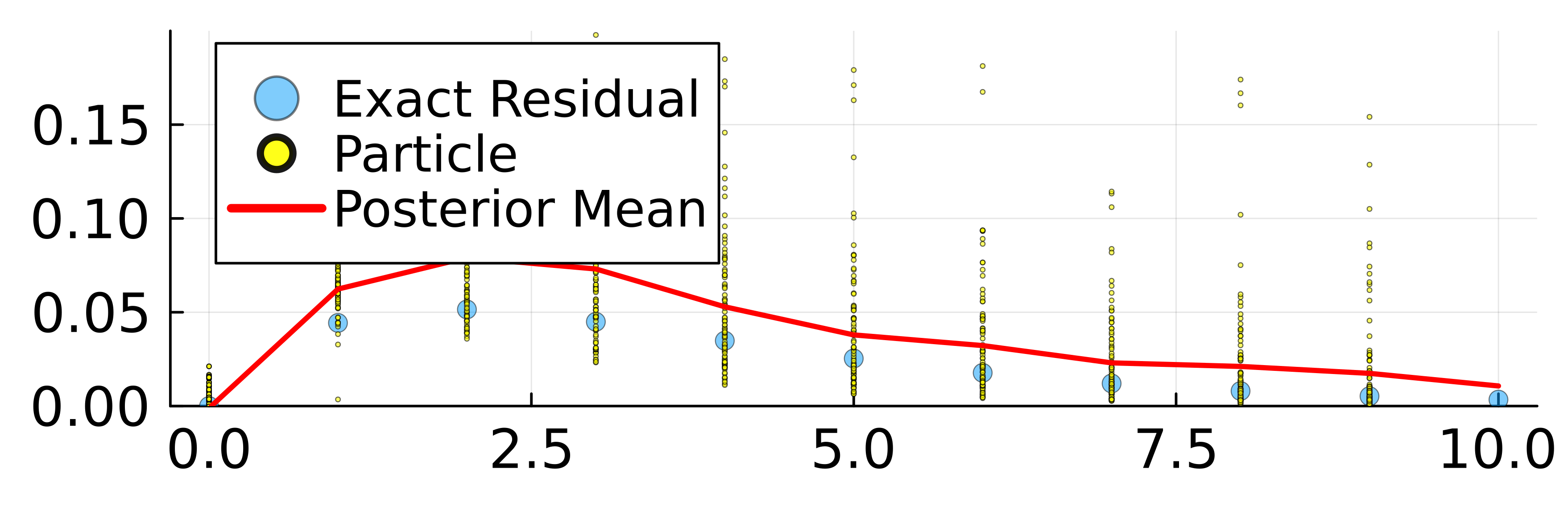}
  \end{subfigure}
  \caption{{Discretization standard deviations for $A(t)$. 
The first component, $\sigma_{1:10}^1$, of samples drawn from the prior distribution $p(\sigma_{1:10})$ (left panel) and the posterior distribution $p(\sigma_{1:10} \mid y^*_{1:10})$ (right panel), where $\sigma_{1:10} = (\sigma_{1:10}^1, \sigma_{1:10}^2)^\mathsf{T}$.}}
\label{fig.chemical_1dim}

  \begin{subfigure}{0.48\textwidth}
    \centering
    \includegraphics[width=\linewidth]{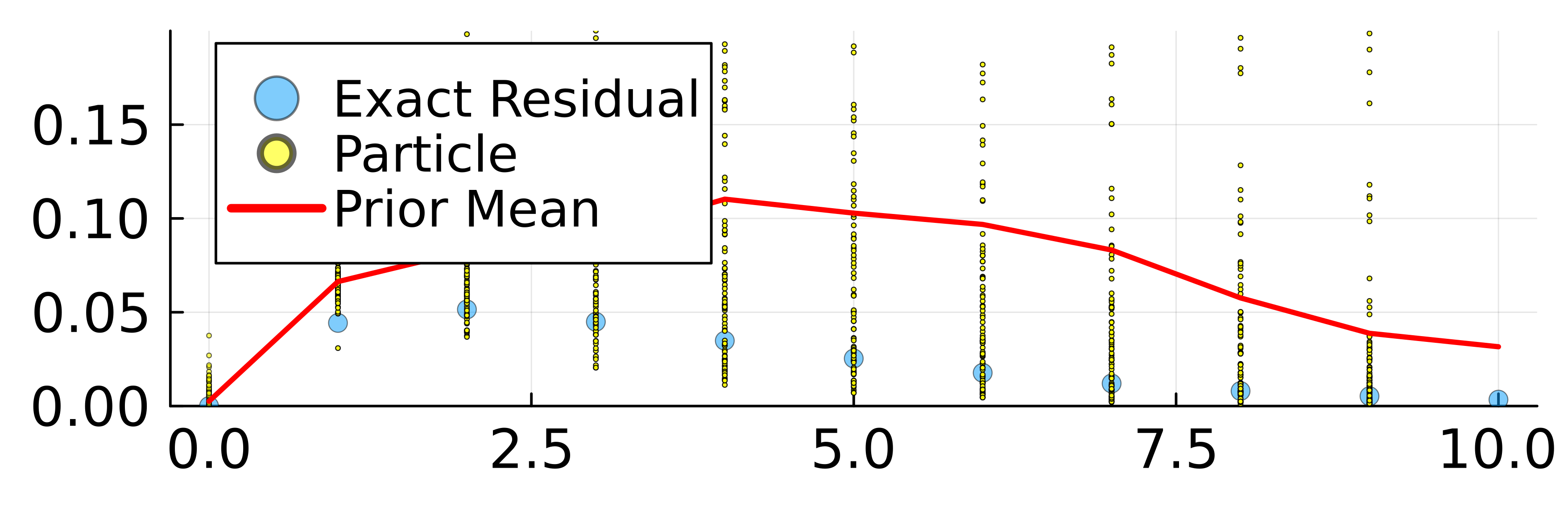}
  \end{subfigure}
  \hfill
  \begin{subfigure}{0.48\textwidth}
    \centering
    \includegraphics[width=\linewidth]{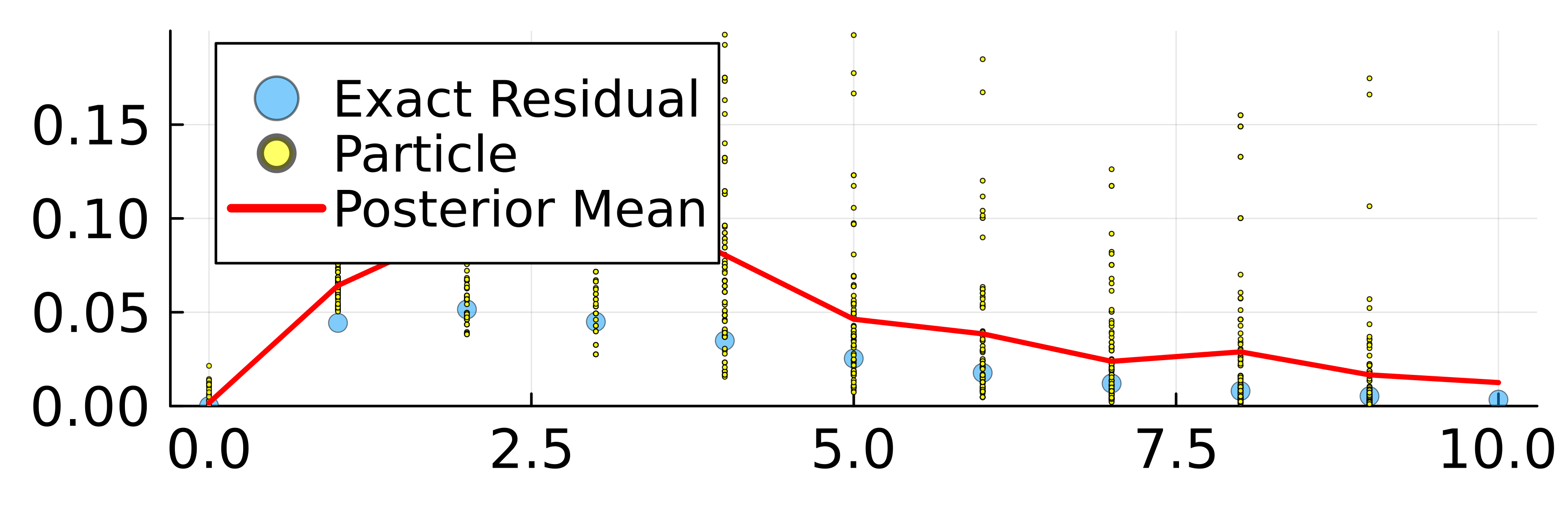}
  \end{subfigure}
  \caption{{Discretization standard deviations for $B(t)$. 
The second component, $\sigma_{1:10}^2$, of samples drawn from the prior distribution $p(\sigma_{1:10})$ (left panel) and the posterior distribution $p(\sigma_{1:10} \mid y^*_{1:10})$ (right panel), where $\sigma_{1:10} = (\sigma_{1:10}^1, \sigma_{1:10}^2)^\mathsf{T}$.}}
\label{fig.chemical_2dim}
\end{figure}

\subsection{Chemical
Reaction Network}

We begin with an evaluation of discretization error quantification for a simple chemical reaction network. We consider the following two-dimensional ODE
\[
\frac{d}{dt}
\begin{bmatrix}
A(t) \\
B(t)
\end{bmatrix}
=
\begin{bmatrix}
- k_1 A(t) + k_2 B(t) \\
k_1 A(t) - k_2 B(t)
\end{bmatrix},
\]
where \(k_1\) and \(k_2\) are reaction rate parameters. We fix $k_1 = 0.5$ and $k_2 = 0.005$, and observations $y^*_{0:10}$ are  given at $t = 0.0, 1.0,...,10.0$, and the covariance matrix is set to $\Gamma = \text{diag}(0.1^2, 0.1^2)$. The number of particles is set to $100$. The distribution $M_{i, j} \sim P$ in the prior \eqref{eq:proposed_prior} is defined by $M_{i, j} = m_{i, j} \cdot I$ where $m_{i, j} \sim \text{Gam} (\alpha, \beta)$. The hyperparameter candidates are selected from $\alpha \in \left\{ 1.1, 1.2,...,2.0 \right\}$ and $\beta \in \left\{ 0.6, ..., 2.0 \right\}$, a total of $150$ parameter candidates. We evaluate the marginal likelihood for each hyperparameter candidate and choose the hyperparameter that maximizes the marginal likelihood.

Figures \ref{fig.chemical_1dim} and \ref{fig.chemical_2dim} illustrate samples (particles) from the prior distribution of the discretization error standard deviations and those from the posterior conditioned on the observations $y^*_{0:10}$, along with the exact discretization errors for comparison. 
In both figures, the left panels show samples drawn from the prior, while the right panels display particles from the posterior obtained using Algorithm \ref{alg:particle}. 
The results indicate that both the prior and posterior capture the trajectory of the discretization error well. 
Moreover, the posterior particles are slightly closer to the observations, reflecting the effect of the inference procedure. 
Importantly, even before filtering, the prior already captures the true behavior of the discretization error, suggesting that the proposed prior provides a reasonable characterization of the discretization error.

\subsection{Pendulum System} 
We {next} consider the pendulum system $y''(t) = (-g/{\theta})\sin y(t)$,
where $g := 9.81 $ denotes the gravitational acceleration.  
This system involves a parameter ${\theta}$ that is to be inferred from observations.  
The system is equivalent to the first-order system
\begin{equation}\label{eq:pendulum_reduce}
\frac{d}{dt}
\begin{bmatrix}
y_1 (t) \\
y_2 (t)
\end{bmatrix}
=
\begin{bmatrix}
y_2 (t) \\[2ex]
-\dfrac{g}{{\theta}} \sin y_1 (t)
\end{bmatrix}.
\end{equation}
We employ the explicit Euler method with step size $h=0.05$. The observation operator $H$ is set to $\mathrm{diag} (3.0, 3.0)$.

\begin{figure}[t] 
  \centering

  \begin{subfigure}[b]{0.49\textwidth}
    \includegraphics[width=\linewidth]{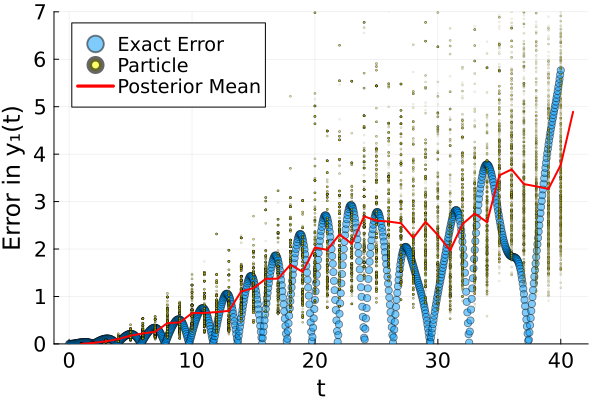}
  \end{subfigure}
  \hfill
  \begin{subfigure}[b]{0.49\textwidth}
    \includegraphics[width=\linewidth]{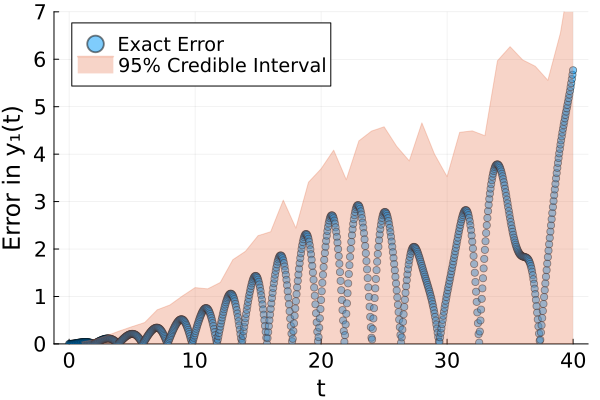}
  \end{subfigure}
   \caption{Discretization error quantification results for $y_1(t)$ in the pendulum system. 
\textbf{Left:} the first component $\sigma_{t_i}^1$ of the particle 
$\sigma_{t_i} = (\sigma_{t_i}^1, \sigma_{t_i}^2)^\mathsf{T}$. 
\textbf{Right:} the $95\%$ credible interval evaluated using samples 
$r_{t_i}^1 \sim \mathcal{N}(0, (\sigma_{t_i}^1)^2)$ drawn from the posterior predictive distribution. }
  \label{fig:Pendulum_particle_noparam}

  \begin{subfigure}[b]{0.49\textwidth}
    \includegraphics[width=\linewidth]{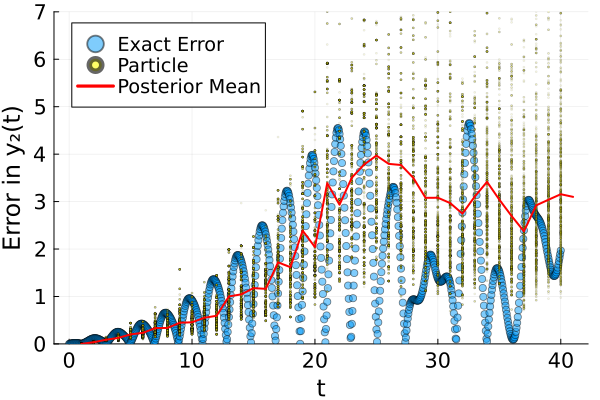}
  \end{subfigure}
  \hfill
  \begin{subfigure}[b]{0.49\textwidth}
    \includegraphics[width=\linewidth]{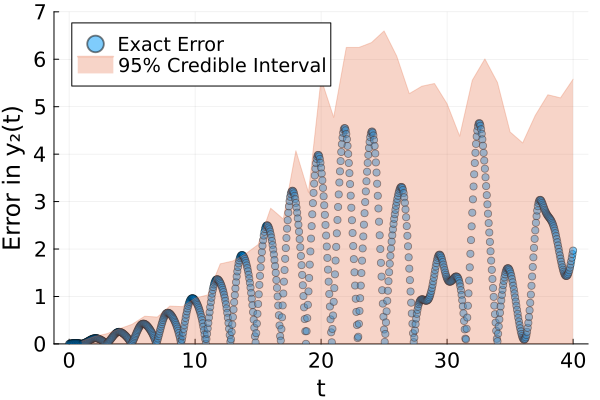}
  \end{subfigure}

  \caption{Discretization error quantification results for $y_2(t)$ in the pendulum system. 
\textbf{Left:} the second component $\sigma_{t_i}^2$ of the particle 
$\sigma_{t_i} = (\sigma_{t_i}^1, \sigma_{t_i}^2)^\mathsf{T}$. 
\textbf{Right:} the $95\%$ credible interval evaluated using samples 
$r_{t_i}^2 \sim \mathcal{N}(0, (\sigma_{t_i}^2)^2)$ drawn from the posterior predictive distribution. }
  \label{fig:Pendulum_posterior_predictive}
\end{figure}

\subsubsection{Discretization Error Quantification}

We first assess the performance of the proposed method (described in Section~\ref{subsec:state-space}) with the parameter ${\theta}$ fixed at ${\theta^*} = 3.0$.
Observations are given at $t=1,2,\dots,40$ with the covariance matrix  given by $\Gamma=\mathrm{diag}(1.0^2, 1.0^2) $. 
The number of particles is set to \( 1000 \), and the distribution \( P \) is fixed to \( m \cdot I \) with \( m \sim \mathrm{Gam}(\alpha, \beta) \).  
The hyperparameters $ (\alpha, \beta) $ are selected from the candidate sets
$
\alpha \in \{ 50.0, 55, 60, \ldots, 800.0 \}$ and 
$\beta \in \{ 0.001, 0.00125, 0.0015, \ldots, 0.033 \}$,
yielding a total of $151 \cdot 129 = 19,479$ candidates.

Figures \ref{fig:Pendulum_particle_noparam} and \ref{fig:Pendulum_posterior_predictive} 
show the discretization error quantification results for $y_1(t)$ and $y_2(t)$ in 
\eqref{eq:pendulum_reduce}, respectively. 
For both cases, the absolute values of the exact discretization errors are plotted 
(``Exact Error'' in the figures). 
The left panels of Figures~\ref{fig:Pendulum_particle_noparam} and  \ref{fig:Pendulum_posterior_predictive} present the first component 
$\sigma_{t_i}^1$ and the second component $\sigma_{t_i}^2$ of the particles 
$\sigma_{t_i} = (\sigma_{t_i}^1, \sigma_{t_i}^2)^\mathsf{T}$ obtained by the proposed method. Recalling that the discretization error variance 
$\Sigma_{t_i}$ for $(y_1(t), y_2(t))^\mathsf{T}$ is expressed as
$
  \mathrm{diag}\!\left( (\sigma_{t_i}^1)^2,\, (\sigma_{t_i}^2)^2 \right)
$
as discussed in~\eqref{eq:assumption_dist_error},
$(\sigma_{t_i}^1)^2$ and $(\sigma_{t_i}^2)^2$ correspond to the discretization error variances 
for $y_1(t)$ and $y_2(t)$, respectively. 
We observe that the particles successfully capture the temporal trend of the exact errors. 
In particular, the decrease in the exact errors around $t = 30$ is accurately reflected via the proposed method. 
Such behavior cannot be captured by previous approaches~\cite{MARUMO2024modelling, Matsuda2019estimation,  miyatake2025quantifying}, which assume that the discretization error 
increases monotonically. Figures \ref{fig:Pendulum_particle_noparam} and \ref{fig:Pendulum_posterior_predictive} also depict the $95\%$ credible interval of the 
discretization errors, drawn from the posterior predictive 
distribution. Specifically, each $r_{t_i}^1 $ and $r_{t_i}^2$ is sampled from normal distributions 
$\mathcal{N}(0, (\sigma_{t_i}^1)^2)$ and $\mathcal{N}(0, (\sigma_{t_i}^2)^2)$, where each 
standard deviation $\sigma_{t_i}$ is randomly selected from the set of particles obtained 
by the proposed method. 
We can confirm that a substantial proportion of the exact discretization errors lie within this credible interval.

We  examine the behavior of the hyperparameters $(\alpha, \beta)$ selected by the empirical Bayesian approach.
Note that under our setting, where $M = m \cdot I$ with $m \sim \mathrm{Gam}(\alpha, \beta)$, we have
\[
\mathbb{E}_{M \sim P_{(\alpha, \beta)}}[M] 
= \mathrm{diag}\!\left(
\mathbb{E}_{m \sim \mathrm{Gam}(\alpha, \beta)}[m],
\ldots,
\mathbb{E}_{m \sim \mathrm{Gam}(\alpha, \beta)}[m]
\right)
= \mathrm{diag}(\alpha \cdot \beta, \ldots, \alpha\cdot \beta).
\]
Table~\ref{tb:oscillation_parameters_top20} lists the top 10 parameter pairs selected from the $19,\!479$ candidates.
We observe that the selected parameters tend to satisfy $\alpha \cdot \beta \approx 1$, which is consistent with the theoretical implications discussed in Section~\ref{sec:prior}.
Figure~\ref{fig:Pendulum_heatmap} shows the heatmap of the log {marginal} likelihood for each parameter candidate.
We find that regions of {high marginal} likelihoods align with the curve $\alpha \cdot \beta = 1$.

\begin{table}[t]
  \centering
  \scriptsize
  \resizebox{0.55\textwidth}{!}{%
  \begin{tabular}{|c|c|c|c|}
    \hline
    Rank & $(\alpha,\beta)$ & $\alpha \cdot \beta$ & {LML} \\ \hline
     1 & $(335.000,\ 0.0030)$ & $1.0050$ & $-241.806424$ \\ \hline
     2 & $(365.000,\ 0.0027)$ & $1.0037$ & $-242.531858$ \\ \hline
     3 & $(670.000,\ 0.0015)$ & $1.0050$ & $-243.266446$ \\ \hline
     4 & $(155.000,\ 0.0065)$ & $1.0075$ & $-243.994101$ \\ \hline
     5 & $(175.000,\ 0.0057)$ & $1.0063$ & $-244.423485$ \\ \hline
     6 & $(445.000,\ 0.0022)$ & $1.0012$ & $-245.456227$ \\ \hline
     7 & $(130.000,\ 0.0077)$ & $1.0075$ & $-245.996058$ \\ \hline
     8 & $(575.000,\ 0.0018)$ & $1.0063$ & $-246.007687$ \\ \hline
     9 & $(310.000,\ 0.0032)$ & $1.0075$ & $-246.112785$ \\ \hline
    10 & $(115.000,\ 0.0088)$ & $1.0063$ & $-246.484862$ \\ \hline
 \end{tabular}%
  }
  \caption{Top 10 parameter pairs $(\alpha, \beta)$ with their corresponding {log marginal likelihoods (LML)} $\log p(y^*_{1:40} \mid \alpha, \beta)$ in the experiment for the pendulum system.}
\label{tb:oscillation_parameters_top20}
\end{table}

\begin{figure}[t]
\centering

\begin{subfigure}{0.495\textwidth}
    \centering
    \includegraphics[width=\textwidth]{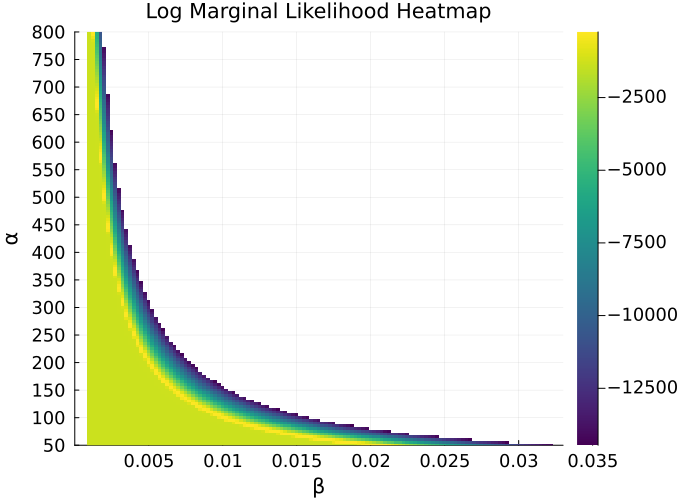}
\end{subfigure}
\hfill
\begin{subfigure}{0.495\textwidth}
    \centering
    \includegraphics[width=\textwidth]{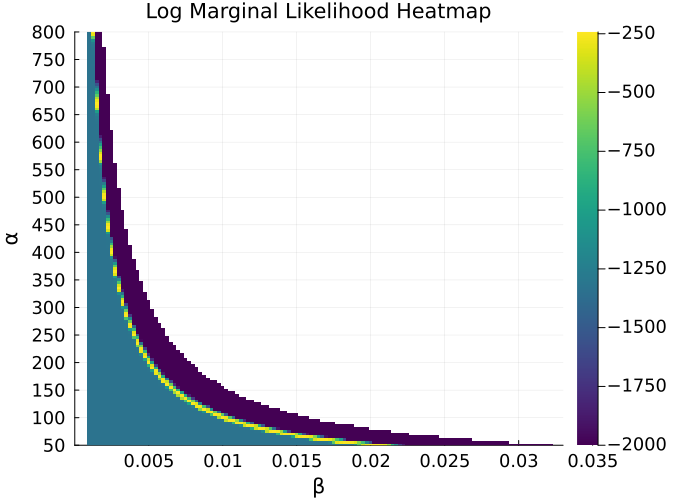}
\end{subfigure}

\caption{
Log {marginal} likelihood (LML) heatmaps in the pendulum system. 
White regions indicate that the likelihood is effectively zero (i.e., output of the log-likelihood is \texttt{-Inf}).
Left: heatmap over all candidate hyperparameter settings. 
Right: heatmap with the color scale fixed below $-2000$ to emphasize higher-likelihood regions.
}
\label{fig:Pendulum_heatmap}
\end{figure}

\subsubsection{Parameter Inference with Discretization Error Quantification}\label{subsec:param_pendulum}

Next, 
we assess the performance of the proposed method (described in Section~\ref{subsec:JointBayes}) 
that jointly estimates the discretization error variance and the ODE parameter.
The observations are generated from the process with the true parameter value ${\theta^*} = 4.0$ at discrete time points $t = 11, \ldots, 40$, with the covariance matrix set to $\Gamma = \text{diag}(2.0^2, 2.0^2)$.
A prior on the ODE parameter $\theta$ is specified as $\mathcal{N}(3.0, 2.0^2)$. 
The number of particles is set to $80{,}000$.  The hyperparameters $(\alpha, \beta)$ are chosen from the pair $\lambda = (\alpha, 1/\alpha)$, where $\alpha \in \{5.0, 10.0, 15.0, \ldots, 500.0\}$.

Figure~\ref{fig:Pendulum_ODEparam_poterior} shows the posterior distributions of the ODE parameter, comparing it with an alternative that does not account for discretization errors. 
The latter is obtained by implementing Algorithm~\ref{alg:particle_self}, with $\sigma_{t_i} = (0.0, 0.0)^\mathsf{T}$ for all $t_i$.   
As observed in the figure, the peaks of both posteriors are slightly shifted to the left of the true parameter ${\theta^*} = 4.0$.  However, the proposed method exhibits broader support, as it accounts for discretization errors.  
In contrast, the posterior that ignores discretization errors behaves like a Dirac measure, indicating that such neglect leads to highly confident yet inaccurate parameter estimates.

Figures~\ref{fig:Pendulum_particle_withparam} and \ref{fig:Pendulum_posterior_predictive_withparam} show discretization error quantification results for $y_1 (t)$ and $y_2 (t)$ respectively. 
These results show that the proposed method underestimates the discretization error. 
This is likely due to the fact that, in the presence of discretization errors, some particles become concentrated around parameter values different from ${\theta^*} = 4.0$ that happen to match the observations.
As a result, the method may incorrectly infer that the discretization error is small,
highlighting a potential limitation of the proposed approach.

\begin{figure}[t] 
  \centering

  \begin{subfigure}[b]{0.49\textwidth}
    \includegraphics[width=\linewidth]{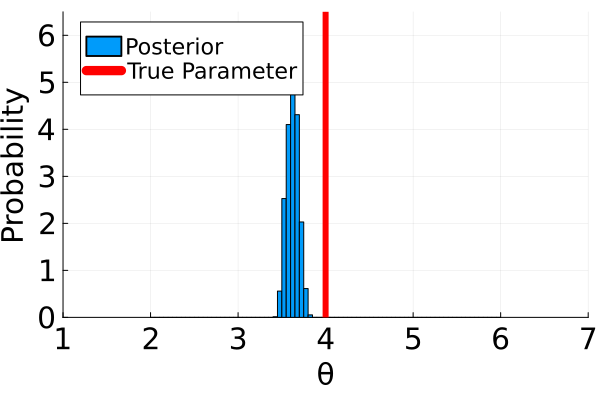}
  \end{subfigure}
  \hfill
  \begin{subfigure}[b]{0.49\textwidth}
    \includegraphics[width=\linewidth]{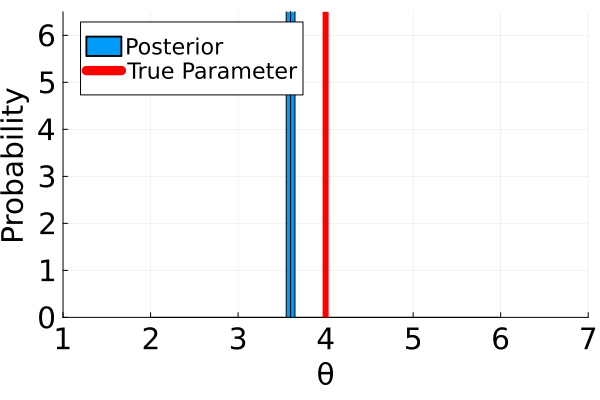}
  \end{subfigure}
\caption{Posterior distributions of the ODE parameter ${\theta}$ for the pendulum system. The left panel shows the posterior obtained using Algorithm \ref{alg:particle_self} with the proposed prior on the discretization error variances, whereas the right panel shows the posterior obtained by applying Algorithm \ref{alg:particle_self} with $\sigma_{t_i} = (0, 0)^\mathsf{T}$ for all $t_i$. } \label{fig:Pendulum_ODEparam_poterior}
\end{figure}

\begin{figure}[t] 
  \centering

  \begin{subfigure}[b]{0.49\textwidth}
    \includegraphics[width=\linewidth]{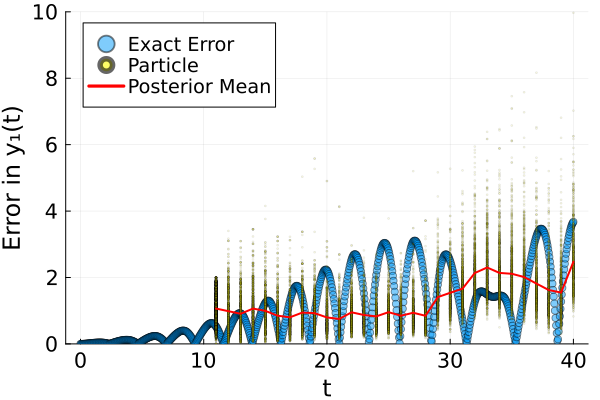}
  \end{subfigure}
  \hfill
  \begin{subfigure}[b]{0.49\textwidth}
    \includegraphics[width=\linewidth]    {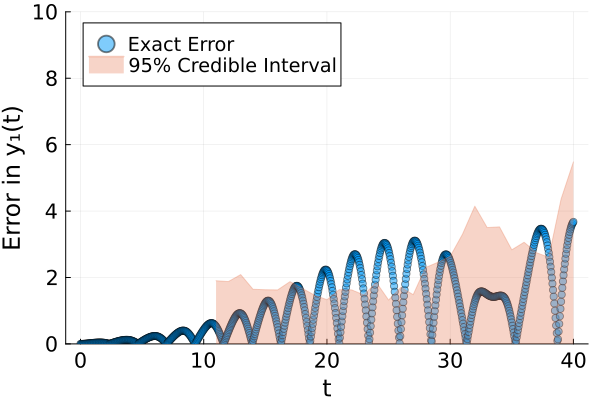}
  \end{subfigure}
   \caption{Discretization error quantification results for $y_1(t)$ in the pendulum system. 
\textbf{Left:} the first component $\sigma_{t_i}^1$ of the particle 
$\sigma_{t_i} = (\sigma_{t_i}^1, \sigma_{t_i}^2)^\mathsf{T}$. 
\textbf{Right:} the $95\%$ credible interval evaluated using samples 
$r_{t_i}^1 \sim \mathcal{N}(0, (\sigma_{t_i}^1)^2)$ drawn from the posterior predictive distribution.}
  \label{fig:Pendulum_particle_withparam}

  \begin{subfigure}[b]{0.49\textwidth}
    \includegraphics[width=\linewidth]{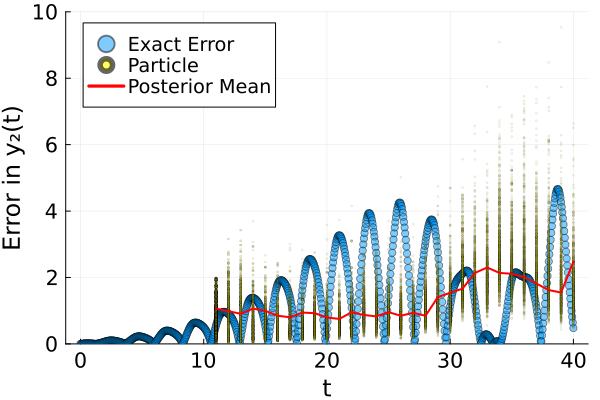}
  \end{subfigure}
  \hfill
  \begin{subfigure}[b]{0.49\textwidth}
    \includegraphics[width=\linewidth]{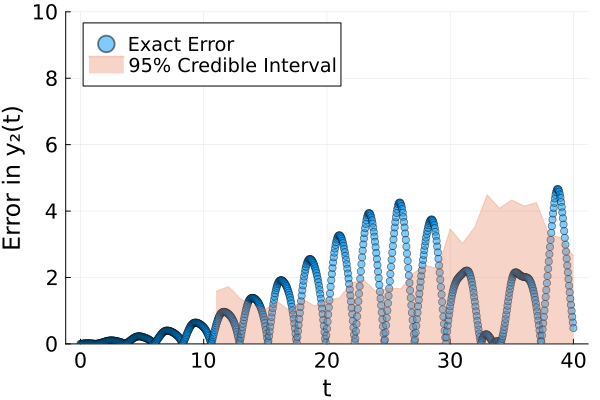}
  \end{subfigure}

  \caption{Discretization error quantification results for $y_2(t)$ in the pendulum system. 
\textbf{Left:} the second component $\sigma_{t_i}^2$ of the particle 
$\sigma_{t_i} = (\sigma_{t_i}^1, \sigma_{t_i}^2)^\mathsf{T}$. 
\textbf{Right:} the $95\%$ credible interval evaluated using samples 
$r_{t_i}^2  \sim \mathcal{N}(0, (\sigma_{t_i}^2)^2)$ drawn from the posterior predictive distribution.}
  \label{fig:Pendulum_posterior_predictive_withparam}
\end{figure}

\subsection{FitzHugh–Nagumo model}

We also demonstrate the effectiveness of the proposed method with the FitzHugh--Nagumo (FN) model:
\begin{equation}
\frac{d}{dt}
\left[
\begin{array}{l}
V \\
R
\end{array}
\right]
=
\left[
\begin{array}{l}
{\theta_3} \left( \rule{0pt}{2.5ex} V - \dfrac{V^3}{3} + R \right) \\[2.ex]
\rule{0pt}{2.5ex} -\dfrac{1}{{\theta_3}} \left( V - {\theta_1} + {\theta_2} R \right)
\end{array}
\right].
\label{eq:fhn}
\end{equation}
The two-dimensional differential equation involves three unknown parameters, $\theta = {(\theta_1, \theta_2, \theta_3)}$, which are to be inferred. 
We employ the explicit Euler method with a step size of $h = 0.2$.  The observation operator $H$ is set to the identity matrix.

\subsubsection{Discretization Error Quantification}
\label{subsec:DEQ}

We assess the performance of the proposed method (described in Section~\ref{subsec:state-space}) with the parameters fixed at ${(\theta_1^*, \theta_2^*, \theta_3^*)} = (0.2,\ 0.1,\ -0.5)$.
The covariance matrix of the observation noise is set to $\Gamma = \mathrm{diag}(0.1^2,\,0.1^2)$.
Observation data $y^*_{t_i}$ are obtained at \(91\) discrete time points,
$t_i = 10, 11, \dots, 100$.
The number of particles is set to $1,000$, and the distribution $P$ is fixed to $m \cdot I$ with $m \sim \mathrm{Gam}(\alpha, \beta) $.  
The hyperparameters $ (\alpha, \beta) $ are selected from the candidate sets
$
\alpha \in \{5.0, 10.0, \dots, 500.0\}$ and 
$\beta \in \{0.002, 0.004, \dots, 1.0\}$,
yielding a total of $100 \cdot 500 = 50,000$ candidates. 
We evaluate the log {marginal} likelihoods $\log p\big(y^*_{10:100} \mid \alpha, \beta \big)$ for $(\alpha, \beta)$ using 50 particles, and choose the pair that achieves the  maximum  {marginal} likelihood.

\begin{figure}[t] 
  \centering

  \begin{subfigure}[b]{0.49\textwidth}
    \includegraphics[width=\linewidth]{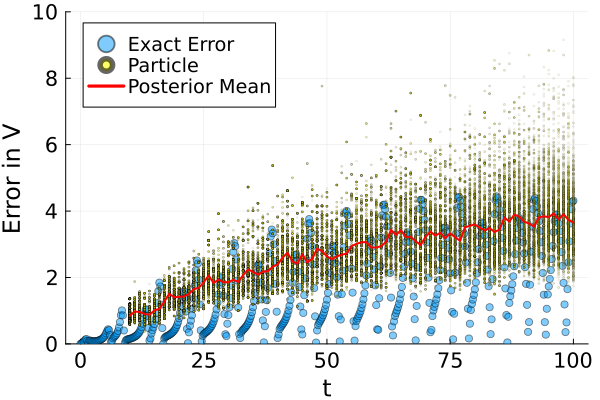}
  \end{subfigure}
  \hfill
  \begin{subfigure}[b]{0.49\textwidth}
    \includegraphics[width=\linewidth]{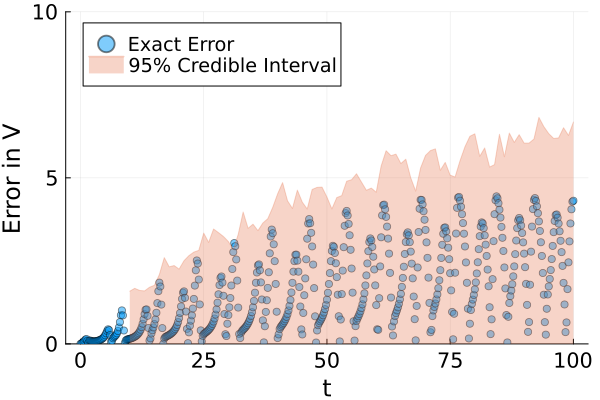}
  \end{subfigure}
   \caption{Discretization error quantification results for $V$ in the FitzHugh--Nagumo model. 
\textbf{Left:} the first component $\sigma_{t_i}^1$ of the particle 
$\sigma_{t_i} = (\sigma_{t_i}^1, \sigma_{t_i}^2)^\mathsf{T}$. 
\textbf{Right:} the $95\%$ credible interval evaluated using samples 
$r^1_{t_i} \sim \mathcal{N}(0, (\sigma_{t_i}^1)^2)$ drawn from the posterior predictive distribution.}
  \label{fig:particle_noparam}

  \begin{subfigure}[b]{0.49\textwidth}
    \includegraphics[width=\linewidth]
    {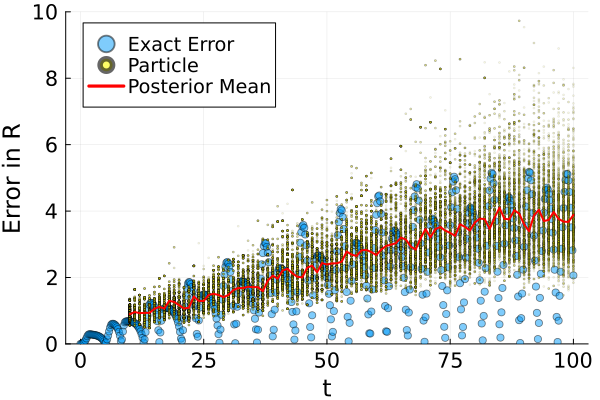}
  \end{subfigure}
  \hfill
  \begin{subfigure}[b]{0.49\textwidth}
    \includegraphics[width=\linewidth]{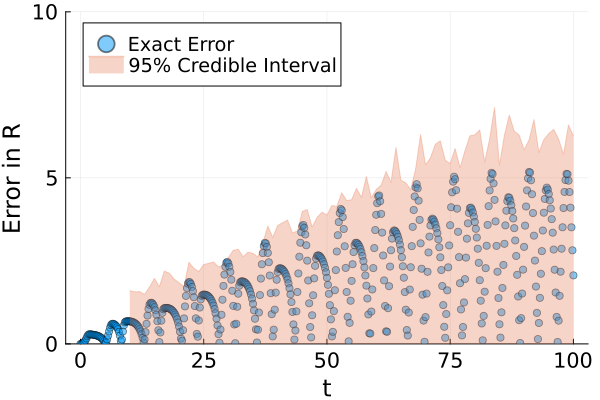}
  \end{subfigure}

  \caption{Discretization error quantification results for $R$ in the FitzHugh--Nagumo model. 
\textbf{Left:} the second component $\sigma_{t_i}^2$ of the particle 
$\sigma_{t_i} = (\sigma_{t_i}^1, \sigma_{t_i}^2)^\mathsf{T}$. 
\textbf{Right:} the $95\%$ credible interval evaluated using samples 
$r^2_{t_i} \sim \mathcal{N}(0, (\sigma_{t_i}^2)^2)$ drawn from the posterior predictive distribution.}
  \label{fig:posterior_predictive}
  \vspace{-5mm}
\end{figure}

Figures~\ref{fig:particle_noparam} and \ref{fig:posterior_predictive} show discretization error {quantification} results for $V$ and $R$ in the FN model, along with the absolute values of the exact discretization errors $r_i$. We can confirm that the temporal evolution of the particles and the  $95\%$ credible intervals closely follows the  time course of the discretization error $r_i$, demonstrating the effectiveness of the proposed approach.

Figure~\ref{fig:heatmap} presents a heatmap of the {log marginal likelihoods} $\log p\!\left( y^*_{10:100} \mid \alpha, \beta \right)$ for all candidate parameter pairs. The heatmap indicates that parameter pairs with {high marginal} likelihood values tend to cluster 
along the curve defined by $\alpha \cdot \beta = 1$, similar to what is observed for the pendulum system.

\begin{figure}[t]
  \centering

  \begin{subfigure}{0.48\textwidth}
    \centering
    \includegraphics[width=\textwidth]{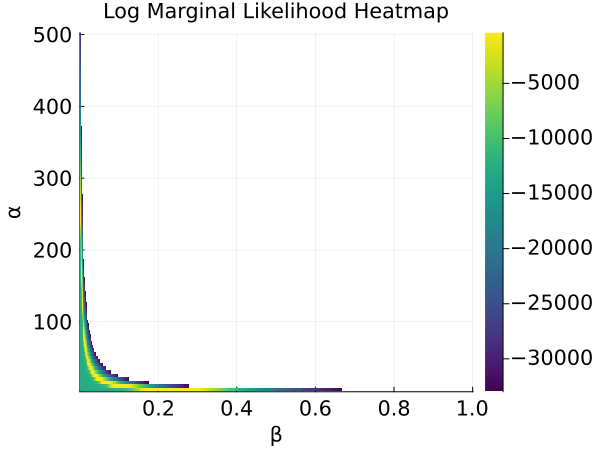}
  \end{subfigure}
  \hfill
  \begin{subfigure}{0.48\textwidth}
    \centering
    \includegraphics[width=\textwidth]{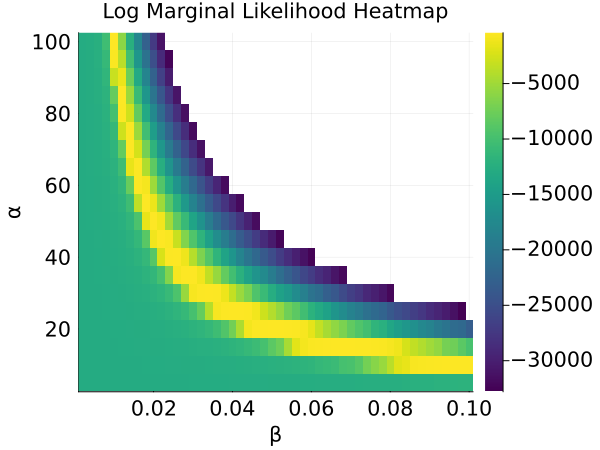}
  \end{subfigure}

  \caption{Log {marginal} likelihood heatmaps in Subsection~\ref{subsec:DEQ}. White regions indicate that the likelihood is effectively zero (i.e., output of the log {marginal} likelihood is \texttt{-Inf}). Left: full hyperparameter range. Right: zoomed region with $\alpha \in \{5,\dots,100\}$ and $\beta \in \{0.002,0.004, \dots,0.10\}$.}
  \label{fig:heatmap}
  \vspace{-6mm}
\end{figure}

\subsubsection{Parameter Inference with Discretization Error Quantification}
\begin{figure}[t]
  \centering

  \begin{subfigure}[b]{0.49\textwidth}
    \includegraphics[width=\linewidth]{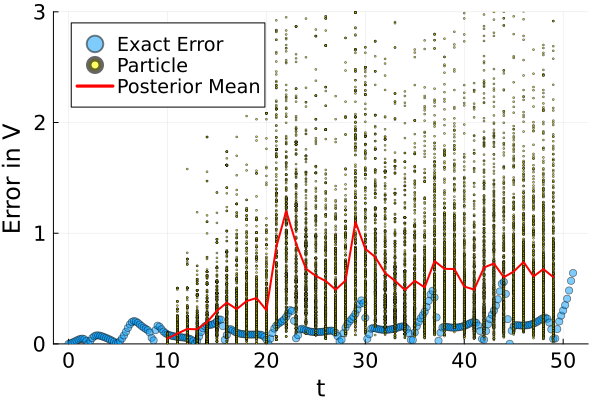}
  \end{subfigure}
  \hfill
  \begin{subfigure}[b]{0.49\textwidth}
    \includegraphics[width=\linewidth]{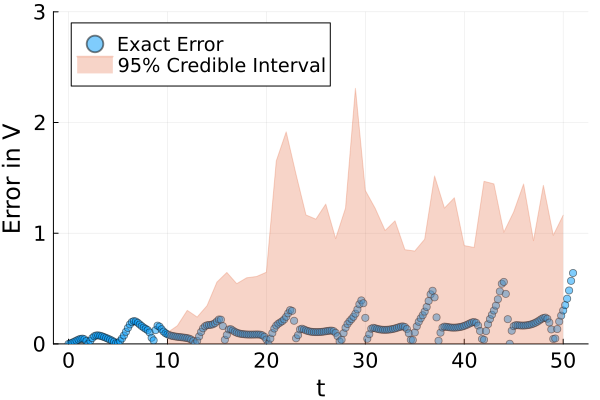}
  \end{subfigure}
  \caption{Discretization error quantification results for $V$ in the FitzHugh--Nagumo model. 
\textbf{Left:} the first component $\sigma_{t_i}^1$ of the particle 
$\sigma_{t_i} = (\sigma_{t_i}^1, \sigma_{t_i}^2)^\mathsf{T}$. 
\textbf{Right:} the $95\%$ credible interval evaluated using samples 
$r^1_{t_i} \sim \mathcal{N}(0, (\sigma_{t_i}^1)^2)$ drawn from the posterior predictive distribution.}
  \label{fig:particle}

  \begin{subfigure}[b]{0.49\textwidth}
    \includegraphics[width=\linewidth]
    {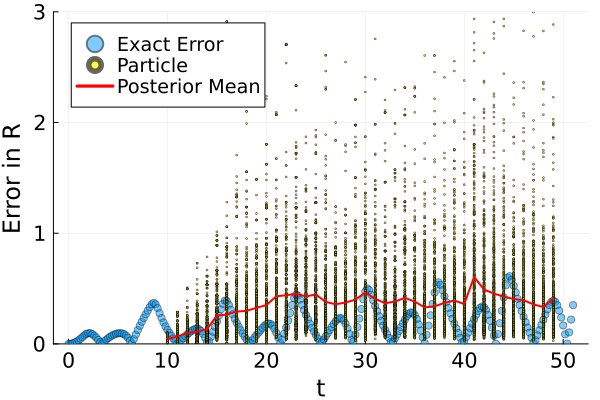}
  \end{subfigure}
  \hfill
  \begin{subfigure}[b]{0.49\textwidth}
    \includegraphics[width=\linewidth]{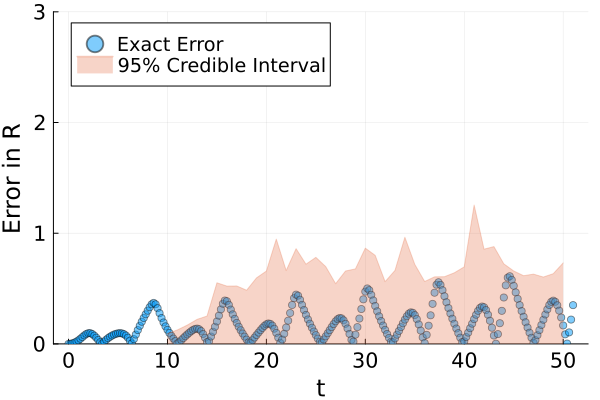}
  \end{subfigure}
  \caption{Discretization error quantification results for $R$ in the FitzHugh--Nagumo model. 
\textbf{Left:} the second component $\sigma_{t_i}^2$ of the particle 
$\sigma_{t_i} = (\sigma_{t_i}^1, \sigma_{t_i}^2)^\mathsf{T}$. 
\textbf{Right:} the $95\%$ credible interval evaluated using samples 
$r^2_{t_i} \sim \mathcal{N}(0, (\sigma_{t_i}^2)^2)$ from the posterior predictive distribution.}
  \label{fig:predictives_discretized_posterior}
\vspace{-5mm}
  \end{figure}
  
We address the simultaneous inference of ODE parameters and discretization errors using the self-organizing technique described in Subsection~\ref{subsec:JointBayes}. The priors for the ODE parameters are set as follows:
\begin{align*}
    {\theta_1}  \sim \mathcal{N}_{[-1.5,\, 1.5]}(0, 0.9^2), ~~~
    {\theta_2}  \sim \mathcal{N}_{[-0.1,\, 1.0]}(1.0, 1.0^2),  ~~~
    {\theta_3}  \sim \mathcal{N}_{[0.1,\, 2.0]}(1.0, 1.0^2),
\end{align*}
where $\mathcal{N}_{[s_1,\, s_2]}(\mu, \sigma^2)$ denotes a truncated normal distribution of $\mathcal{N}(\mu, \sigma^2)$ restricted to the interval $[s_1, s_2]$. The number of particles is fixed at $500{,}000$, which is larger than in the previous subsection to account for the increased dimensionality of the latent space. The noise matrix $\Gamma$ in the observation process is set to $\mathrm{diag}(1.0^2,\ 1.0^2)$. Observations $y^*_{t_i}$ are assumed to be obtained at time points $t_i = 10, 11, \ldots, 50$, with the true parameter fixed at ${(\theta_1^*, \theta_2^*, \theta_3^*)} = (0.5,\ 0.2,\ 1.0)$. The distribution $P_\lambda$ is $m \cdot I \sim P_\lambda$, where $m \sim \text{Gam}(\alpha, \beta)$, as in the previous sections.  The hyperparameters $(\alpha, \beta)$ are selected from the pair \( \lambda = (\alpha, \frac{1}{\alpha}) \), such that \( \alpha \cdot \frac{1}{\alpha} = 1 \), where  \( \alpha \in \{2.0, 4.0, 6.0, \ldots, 50.0\} \).

Figures~\ref{fig:particle} and \ref{fig:predictives_discretized_posterior} show discretization error quantification results. It can be seen that the discretization error in \( V \) is overestimated, particularly around time \( t = 20 \). This overestimation is likely due to the local error \( L(t_i) \) also being overestimated in the interval {$t \in [20, 30]$.}
After \( t = 30 \), the particles are more widely spread; however, on average,
they are located close to the exact error. Regarding the discretization error variance in \( R \), we can see that the proposed method successfully captures the exact discretization errors.

Figure~\ref{fig_ODE_posterior} presents the posterior distributions of the ODE parameters obtained using the proposed method, compared with those obtained by ignoring discretization errors. The latter is computed in the same manner as in the pendulum system. The results show that ignoring discretization errors yields posterior distributions that fail to capture the true parameters, particularly $\theta_1$ and $\theta_3$. In contrast, when the discretization error variances are taken into account, the supports of the posterior distributions shift toward the true parameter values.

\begin{figure}[t]
  \centering
  \begin{subfigure}[b]{1.\textwidth}
    \includegraphics[width=\linewidth]{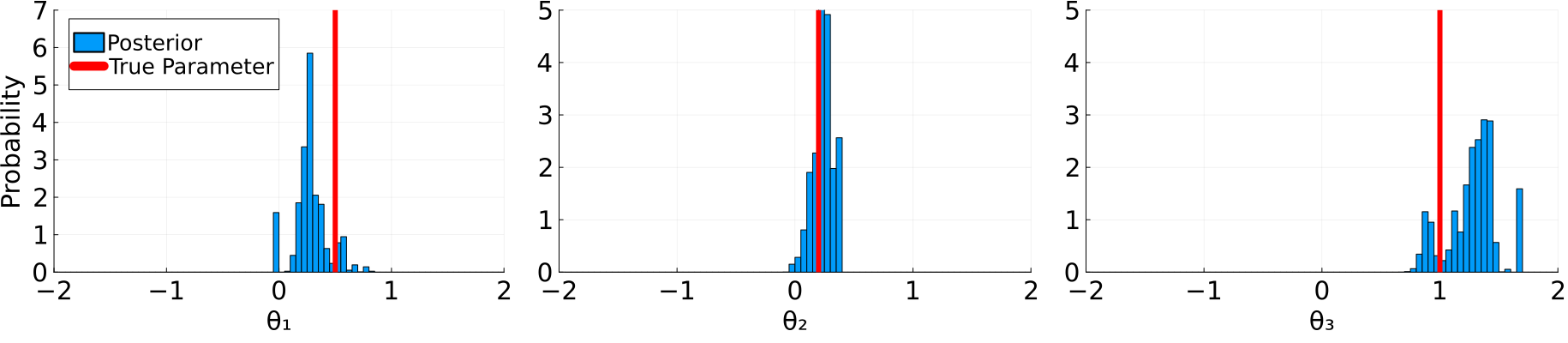}
  \end{subfigure}
  \\
  \vspace{5mm}
  \begin{subfigure}[b]{1.\textwidth}
    \includegraphics[width=\linewidth]{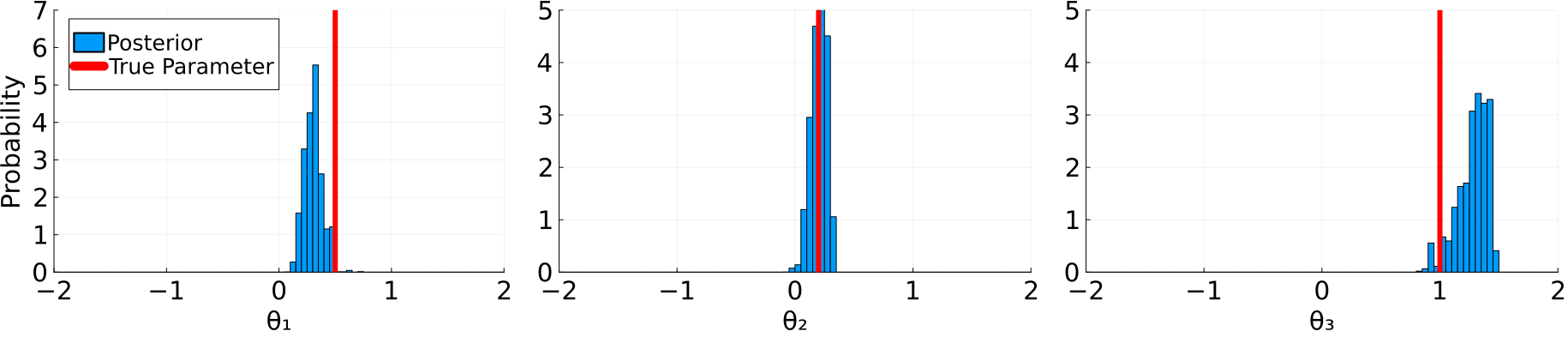}
  \end{subfigure}
  \caption{Posterior distributions of the ODE parameters ${(\theta_1, \theta_2, \theta_3)}$ for the FitzHugh-Nagumo model. The first row shows the posterior obtained using Algorithm \ref{alg:particle_self} with the proposed prior on the discretization error variances, whereas the second row shows the posterior obtained by applying Algorithm \ref{alg:particle_self} with $\sigma_{t_i} = (0, 0)^\mathsf{T}$ for all $t_i$. }
  \label{fig_ODE_posterior}
  \vspace{-5mm}
\end{figure}

\section{Discussion}
\label{sec:discussions}
\subsection{Conclusion}
\label{subsec:conclusions}

In this study, we proposed a Bayesian framework for parameter estimation in ordinary differential equation (ODE) models that explicitly accounts for discretization errors introduced by numerical solvers. 
Our approach models these discretization errors as random variables and performs joint Bayesian inference on both the ODE parameters and the corresponding discretization error variances. 
{A key ingredient of the proposed framework is the introduction of a Markov prior on the discretization error variances. 
This Markov prior allows us to encode a fundamental principle from numerical analysis---that global discretization errors accumulate from local errors---directly into the prior structure. 
As a result, the framework establishes a principled connection between Bayesian inference and numerical analysis, which is difficult to realize in previous discretization error variance approaches.}
For inference, we formulate the resulting model within a state-space framework.
We established a convergence rate for the proposed prior as the solver step size tends to zero, thereby providing a theoretical guarantee. 
The effectiveness of the proposed approach is demonstrated through numerical experiments on the pendulum system and the FitzHugh--Nagumo model, showing improved parameter estimation and uncertainty quantification.

It is worth noting that our joint inference essentially suffers from \emph{non-identifiability}. That is, for a model parameterized by ODE parameters and discretization error variances, multiple combinations of these quantities can yield similar likelihood values. For example:
\begin{itemize}
    \item When both the ODE parameters and the discretization error variance correctly reflect the true underlying parameter and numerical error, the likelihood is expected to be appropriately maximized.    
    \item However, there may exist incorrect parameter values for which the numerical solution coincidentally matches the observed data. In such cases, a combination of incorrect parameters and a small discretization error variance may also lead to an undesired maximization of the likelihood.
\end{itemize}
In general, for inference problems involving non-identifiability, a natural approach is to adopt a Bayesian framework: while the parameters cannot be uniquely identified, it becomes possible to infer plausible parameter values with uncertainty quantification, through the use of appropriate priors.
In our setting, the use of a Bayesian framework for the joint inference problem, together with the introduction of a prior derived from standard error analysis, can be interpreted as a principled way to mitigate the effects of non-identifiability.

\subsection{Further Work}
\label{subsec:furtherworks}
Several directions remain for future work. 
First, the proposed prior assumes diagonal discretization error variances, thereby ignoring correlations across state dimensions. This assumption can be reasonable when the dynamics are only weakly coupled. For example, in the case of PDEs where the system states take continuous values and observations are sparse, correlations at the observation points are typically limited, making a diagonal error covariance a plausible assumption.
In more strongly coupled or stiff systems, however, discretization errors may exhibit cross-dimensional correlations. In such cases, the diagonal approximation may underestimate the propagation of uncertainty across states. When this assumption is violated, the method may fail to fully capture the behavior of correlated discretization errors, potentially affecting inference and uncertainty quantification.
Developing an appropriate Markov prior for the non-diagonal case---for example, one based on the Wishart distribution---remains an important open problem.

Second, diagnostic guidance for quantifying discretization error should be established. For instance, the experimental results in Section~\ref{subsec:param_pendulum} suggest that the proposed method may underestimate the discretization error when parameter values that closely fit the observations are present. One possible diagnostic approach to address this issue is to perform experiments with different step sizes and examine how the estimated discretization error changes. Since the behavior of discretization error is often predictable from the order of the numerical solver, deviations of the estimated error from the expected scaling may indicate potential issues, such as a mismatch in the prior or bias in the inference toward incorrect parameter values.

It is also important to extend the proposed method to high-dimensional systems, including partial differential equations. 
Although particle filters are attractive due to their minimal modeling assumptions, they often do not scale well in high-dimensional settings. In PF, the required number of simulations $K$ per time step typically scales exponentially with the state dimension. One promising alternative is to employ the ensemble Kalman filter (EnKF) \cite{evensen2009data}. While both PF and EnKF have linear computational complexity $\mathcal{O}(K)$ with respect to the number of simulations $K$, EnKF can operate with a much smaller simulation cost, typically on the order of $K = 10$--$100$ in many practical applications. However, this  efficiency comes at the cost of assuming a linear Gaussian observation model, 
which limits its direct applicability to the present nonlinear and non-Gaussian setting. 
Developing an alternative modeling framework to address these limitations is therefore an important direction for future research.

Another important direction for future work is the estimation of discretization error in settings where additional sources of discrepancy---such as ODE model misspecification or observational bias---also contribute to the residual between numerical simulations and observational data. Throughout this paper, for simplicity, we assume that both the simulation model and observation noise are well-specified, although this assumption is often violated in real-world scenarios. Regarding observation noise, it can be inferred as a hyperparameter, and in many cases, it can also be partially specified through  observation simulations \cite{quigley2010sub, schalk2007decoding}. Concerning model misspecification, a classical remedy is to explicitly model the degree of misspecification \cite{alvarez2013linear, kennedy2001bayesian}, which can also address the problem considered here. Another possible remedy is to employ generalized Bayesian inference (GBI) schemes \cite{bissiri2016general}. GBI performs robust Bayesian inference for misspecified models by replacing the likelihood function with a cost function that evaluates the goodness of fit of parameters relative to observations. These approaches are known to perform well in various statistical inference problems \cite{dellaporta2022robust, martin2024outlier, gao2023generalized} and are expected to be effective for estimating the discretization error variance in misspecified ODE models.  

Finally, integrating the proposed framework with the field of simulation-based inference (SBI) offers another compelling direction.
SBI has emerged as a powerful approach for parameter inference through model simulations \cite{cranmer2020frontier}. In particular, recent advances in deep neural network–based methods have greatly expanded its applicability to high-dimensional simulators \cite{papamakarios2016fast, radev2023jana, dax2023flow} or simulators constructed by differential equations \cite{gloeckler2025compositional, hikida2025multilevelneuralsimulationbasedinference}.
However, discretization errors introduced by numerical solvers have not yet been incorporated into SBI methodologies.
Developing SBI methodologies that explicitly account for such errors is an important future challenge.

\appendix

\section{Derivation of Smoothing in Particle Filters}\label{appendix_smoothing}

In this appendix, we provide the mathematical background for the simultaneous resampling technique~\eqref{eq_smoothing}, used to obtain smoothed particles
$
\big\{ \big( \Sigma_{0|0}^k, \dots, \Sigma_{i|0:i}^k \big) \big\}_{k=1}^K
$
that represent the distribution \( p \left(\Sigma_{0:i} \mid y^*_{0:i}\right) \), which plays a central role in the present work. 
This technique was originally proposed by Kitagawa~\cite{kitagawa1996monte}.

We establish the result by mathematical induction.
Assume that we already have smoothed particles 
$
\big\{ \big( \Sigma_{0|0}^k, \dots, \Sigma_{i|0:i}^k \big) \big\}_{k=1}^K
$
representing the distribution \( p \left(\Sigma_{0:i} \mid y^*_{0:i}\right) \).
Then, by the prediction step
$
\Sigma_{i+1 \mid 0:i}^k \sim p\big( \Sigma_{t_{i+1}} \mid \Sigma_{i \mid 0:i}^k \big)
$,
we obtain particles
$
\big\{ \big( \Sigma_{0|0}^k, \dots, \Sigma_{i|0:i}^k, \Sigma_{i+1|0:i}^k \big) \big\}_{k=1}^K
$
that approximate the distribution $ p \big(\Sigma_{0:i+1} \mid y^*_{0:i}\big) $, one step ahead in time.

From straightforward computation, we have
\begin{align}
p\left( \Sigma_{0:{i+1}} \mid y_{0:i+1} \right) 
&= \frac{p(y_{i+1} \mid \Sigma_{0:{i+1}}, y_{0:i}) \, p( \Sigma_{0:{i+1}} \mid y_{{0:i}})}
{ p(y_{i+1} \mid y_{0:i}) } \notag \\
&\propto p(y_{i+1} \mid \Sigma_{0:{i+1}}, y_{{0:i}}) \, p( \Sigma_{0:{i+1}} \mid y_{0:i}) \notag \\
&= p(y_{i+1} \mid \Sigma_{{i+1}}) \, p( \Sigma_{0:{i+1}} \mid y_{0:i}),
\label{eq:correction_step_smoothed}
\end{align}
where the final equality follows from 
 the hidden Markov structure, that assumes conditional independence between $y_{i+1}$ and $(\Sigma_{0:i}, y_{{0:i}})$ given $\Sigma_{i+1}$. 

Recall that we have a particle approximation
$
\big\{ \big( \Sigma_{0|0}^k, \dots, \Sigma_{i|0:i}^k, \Sigma_{i+1|0:i}^k \big) \big\}_{k=1}^K
$
for the distribution \( p \left(\Sigma_{0:i+1} \mid y^*_{0:i}\right) \).
Substituting this approximation into \eqref{eq:correction_step_smoothed} gives
\[
p\left( \Sigma_{0:{i+1}} \mid y_{0:i+1}^* \right) 
\propto p(y_{i+1}^* \mid \Sigma_{{i+1}}) \, p( \Sigma_{0:{i+1}} \mid y_{0:i}^* )
\approx \sum_{k=1}^K w_k \, 
\delta_{( \Sigma_{0|0}^k, \dots, \Sigma_{i|0:i}^k, \Sigma_{i+1|0:i}^k )},
\]
where the weights $w_k$ are defined as in \eqref{eq:particle_weight}. 
By applying the same argument as in the correction step of 
Subsection~\ref{subsec:state-space}, 
the particles
$
\big\{ \big( \Sigma_{0|0}^k, \dots, \Sigma_{i|0:i+1}^k \big) \big\}_{k=1}^K
$
obtained by the resampling step (\ref{eq_smoothing}) 
yield a particle approximation for 
$ p \left(\Sigma_{0:i+1} \mid y^*_{0:i+1}\right) $.

\section{Proof of Theorem \ref{theo_convergence_rate}}\label{appendix_proof_asymot}

In this appendix, we prove Theorem \ref{theo_convergence_rate}.
While \eqref{eq:prior_stepsize} defines the Markov prior only at the observation points 
$t_0, \ldots, t_N$, the recursive construction in \eqref{eq:proposed_prior} naturally allows us 
to extend this definition to finer time grids as follows:
\begin{align*}
    p_h (\sigma_{t_{0, 0}},\sigma_{t_{0, 1}},...,\sigma_{t_{0, k-1}},\sigma_{t_{1, 0}}   \ldots, \sigma_{t_{N, 0}}) & {=p_h (\sigma_{t_{0, 0}},\sigma_{t_{0, 1}},...,\sigma_{t_{0, k}},\sigma_{t_{1, 1}}   \ldots, \sigma_{t_{N-1, k}})}  \\
   & = p_h (\sigma_{t_{0, 0}}) \prod_{i=0}^{N-1} \cdot \prod_{j=0}^{k-1} p_h (\sigma_{t_{i, j+1}} \mid \sigma_{t_{i, j}}),
\end{align*}
where  we identify $t_{i,k} = t_{i+1,0}$.

\begin{proof}[Proof of Theorem \ref{theo_convergence_rate}]
Let us recall that random variables $\{ X_h \}_{h>0}$ and $\{ Y_h \}_{h>0}$ indexed by $h$ satisfy $X_h  = \mathcal{O}_p (Y_h)$ if $X_h/Y_h$ is uniformly tight:
$$ \lim_{\epsilon \rightarrow \infty} \sup_{h}\mathbb{P}\left( \left\|\frac{X_h}{Y_h} \right\| \geq \epsilon \right)  = 0. $$
By Markov inequality, we have
\begin{equation}
    \sup_{h} \mathbb{P}\left(\left\|\frac{\sigma_{t_i} }{h^{\min({a, b})}} \right\| \geq \epsilon \right) \leq \sup_{h} \frac{1}{\epsilon\cdot h^{\min({a, b})}} \mathbb{E}_{\sigma_{t_i} \sim p_h (\sigma_{t_i})}[\| \sigma_{t_i}  \|].
\end{equation}
Hence, it suffices to prove that $\mathbb{E}_{p_h (\sigma_{t_i})}[\| \sigma_{t_i}  \|] = \mathcal{O}(h^{\min(a, b)})$ for $i \geq 1$. 

Define \( e_{i, j} := \mathbb{E}_{p_h(\sigma_{t_{i, j}})}[\| \sigma_{t_{i, j}} \|] \) 
to simplify the notation.
By the definition \eqref{eq:proposed_prior},
it holds that for $0 \leq i \leq (N -1)$ and $0 \leq j \leq k-1$, 
\begin{equation*}
\sigma_{t_{i, j+1}  }  = M_{i, j} \cdot \sigma_{t_{i, j} }  + |L\left(t_{i, j}  \right)|.
\end{equation*}
This leads to the inequality 
\begin{align}
    e_{i, j+1} &=\mathbb{E}_{p_h(\sigma_{i, j+1})}[\| \sigma_{i, j+1} \|]  \notag \\  %\notag
    &= \mathbb{E}_{p_h(\sigma_{t_{i, j}}), P_\lambda}[\| M_{i, j} \cdot \sigma_{t_{i, j}} + |L(t_{i, j})| \|]\notag \\[1ex]
    &\leq \mathbb{E}_{p_h(\sigma_{t_{i, j}}), P_\lambda}[\| M_{i, j} \cdot \sigma_{t_{i, j}}  \|] + \| L(t_{i, j}) \| \notag\\[1ex]
    &\leq \mathbb{E}_{p_h(\sigma_{t_{i, j}}), P_\lambda}[\|M_{i, j} \|_{\mathrm{op}} \cdot \|\sigma_{t_{i, j}}  \|] + \| L(t_{i, j}) \| \notag\\[1ex]
    &= \mathbb{E}_{P_\lambda}[\|M_{i, j} \|_{\text{op}} ] \cdot \mathbb{E}_{p_h(\sigma_{t_{i, j}})}[\|\sigma_{t_{i, j}}  \|] +  \| L(t_{i, j}) \|   \label{eq:convergence rate_indepassump}  \\[1ex]
    &=  \mathbb{E}_{P_\lambda} [\| I - ( I  - M_{i, j} )\|_{\text{op}}] \cdot e_{i, j}  +  \| L(t_{i, j}) \| \notag  \notag \\[1ex]
    &\leq \left\{ \mathbb{E}_{P_\lambda} [\| I \|_{\text{op}}] + \mathbb{E}_{P_\lambda} [\|  I  - M_{i, j} \|_{\text{op}}] \right\} \cdot e_{i, j} +  \| L(t_{i, j}) \| \notag   \\
    &\leq (1 + C_M h) \cdot e_{i, j} + \| L(t_{i, j}) \| \label{eq:convergence rate1}  \\
    &\leq (1 + C_M h) \cdot e_{i, j} + C_L h^{{a} +1}  \label{eq:convergence rate2} 
\end{align}
Here, the equality in \eqref{eq:convergence rate_indepassump} follows from Condition (IV), while the inequalities in \eqref{eq:convergence rate1} and \eqref{eq:convergence rate2} follow from Conditions (I) and (II), respectively.

Applying \eqref{eq:convergence rate2} recursively, we obtain
\begin{align}
    e_{i, j+1} 
    &\leq (1+C_M h)^{j+1+ik}\cdot e_{0, 0} + C_L h^{{a} +1} \sum_{l=1}^{j+ik+1} (1+C_Mh)^{l-1}\notag  \\
    &= (1+C_M h)^{j+1+ik} \cdot e_{0, 0} + C_L h^{{a} +1} \cdot  \frac{(1+C_M h)^{j+1+ik} - 1}{C_M h}  \notag \\
     &\leq \exp\big( ( j+1+ik ) C_M h \big) \cdot  e_{0, 0} + C_L h^{{a} } \cdot  \frac{\exp\big( ( j+1+ik ) C_M h \big) - 1}{C_M} \notag \\
          &\leq \exp\big( C_M \left(\left( j+1+ik \right)h\right) \big) \cdot  e_{0, 0} + C_L  h^{{a} } \cdot  \frac{\exp\big( C_M \left( \left( j+1+ik \right)h\right) \big) - 1}{C_M } \label{eq:convergence_rate_l>1}. 
          \\
          &\leq \exp\big( C_M \big(t_{i, j+1} - t_{0} \big) \big) \cdot  e_{0, 0} + C_L h^{{a} } \cdot  \frac{\exp\big( C_M \left(t_{i, j+1} - t_{0} \right)     \big) - 1}{C_M } . 
 \end{align}
 Since this {inequality} holds for  $j = k -1 $ and $0 \leq i \leq N-1 $, we have 
 \begin{align}
  \mathbb{E}_{p_h (\sigma_{t_{i}})}[\| \sigma_{t_i}  \|] &= \mathbb{E}_{p_h (\sigma_{t_{i-1, k}})}[\| \sigma_{t_{i-1, k}}  \|] = e_{i-1, k} = e_{i-1, (k-1) +1}  \notag \\
          & \leq \exp\big( C_M \big( t_{i-1, k} - t_0 \big) \big) \cdot  e_{0, 0} + C_L h^{{a} } \cdot  \frac{\exp\big( C_M \big( t_{i-1, k} - t_0 \big) \big) - 1}{C_M }   \\
          & \leq \exp\big( C_M \left( t_{i} - t_0 \right) \big) \cdot  e_{0, 0} + C_L h^{{a} } \cdot  \frac{\exp\big( C_M ( t_{i} - t_0 ) \big) - 1}{C_M } \label{eq:convergence_rate_final}.
 \end{align}
Note that \( e_{0, 0} = \mathbb{E}_{p_h (\sigma_{t_0})}\!\left[\| \sigma_{t_0} \|\right] = \mathcal{O}(h^{{b}}) \) (Condition (III)), and therefore the first term in \eqref{eq:convergence_rate_final} is also of order \(\mathcal{O}(h^{{b}})\).  
This implies that \eqref{eq:convergence_rate_final} is of order \(\mathcal{O}(h^{\min({a, b})})\), which concludes the proof.
\end{proof}

\section*{Acknowledgement}
This work is supported by JSPS KAKENHI Grant Numbers 24K02951, 24K00540, 25H00449, 24K20750, 25K21806, JP25H01454 and JST ACT-X, Japan, Grant Number JPMJAX25CH.

\bibliographystyle{siamplain}
\bibliography{references}

\end{document}